\documentclass[11pt]{article}
\usepackage{amsmath,amssymb,amsfonts,amsthm,epsfig,verbatim}
\usepackage{times,framed}
\usepackage{color}
\newif\ifhyper\IfFileExists{hyperref.sty}{\hypertrue}{\hyperfalse}
\hypertrue
\ifhyper\usepackage{hyperref}\fi

\def\confversion{0}

\newtheorem{theorem}{Theorem}

\newtheorem{lemma}[theorem]{Lemma}

\newtheorem{corollary}[theorem]{Corollary}
\newtheorem{claim}[theorem]{Claim}
\newtheorem{fact}[theorem]{Fact}

\newtheorem{definition}[theorem]{Definition}
\newtheorem{remark}[theorem]{Remark}

\newcommand{\eps}{\epsilon}
\newcommand{\epsa}{\epsilon'}

\newcommand{\etal}{{\em et al.\ }}

\newcommand{\Var}{\operatorname{Var}}
\newcommand{\reg}{\mathrm{reg}}
\newcommand{\Cov}{\operatorname{Cov}}
\newcommand{\sgn}{\mathrm{sign}}
\newcommand{\lab}{\mathrm{label}}
\newcommand{\sign}{\mathrm{sign}}

\newcommand{\ignore}[1]{}

\newcommand{\cref}[1]{Corollary~\ref{cor:#1}}

\newcommand{\bits}{\{-1,1\}}
\newcommand{\bn}{\bits^n}

\newcommand{\R}{{\mathbb{R}}}

\newcommand{\Z}{{\mathbb Z}}
\newcommand{\E}{\operatorname{{\bf E}}}

\newcommand{\littlesum}{\mathop{\textstyle \sum}}

\newcommand{\poly}{\mathrm{poly}}

\newcommand{\Inf}{\mathrm{Inf}}

\newcommand{\eqdef}{\stackrel{\textrm{def}}{=}}

\renewcommand{\Pr}{\operatorname{{\bf Pr}}}

\newcommand{\nnew}[1]{{\color{blue} #1}}

\newcommand{\dk}{d_{\mathrm K}}

\newcommand{\drect}{d_{{\mathrm{rect}}}}

\newcommand{\green}[1]{{\color{green} {#1}}}
\newcommand{\new}[1]{{\color{red} {#1}}}

\makeatletter
\renewcommand{\subsection}{\@startsection{subsection}{2}{0pt}{-6pt}{-5pt}{\normalsize\bf}}
\renewcommand{\subsubsection}{\@startsection{subsubsection}{3}{0pt}{-12pt}{-5pt}{\normalsize\bf}}
\makeatother

\newcommand{\rnote}[1]{\footnote{{\bf [[Rocco: {#1}\bf ]] }}}

\newcommand{\NC}{\mathsf{NC}}

\newcommand{\Thr}{\mathsf{Thr}}
\newcommand{\AND}{\mathsf{AND}}
\newcommand{\proj}{\mathrm{Proj}}
\newcommand{\res}{\mathrm{Res}}
\newcommand{\head}{\mathrm{Head}}
\newcommand{\tail}{\mathrm{Tail}}
\newcommand{\quadtail}{\mathrm{QuadTail}}

\date{}

\title{
Deterministic Approximate Counting for \\
Juntas of Degree-$2$ Polynomial Threshold Functions
}

\author{Anindya De\thanks{{\tt anindya@math.ias.edu}.  Research supported by Umesh Vazirani's Templeton Foundation Grant 21674.}\\
Institute for Advanced Study\\
\and
Ilias Diakonikolas\thanks{{\tt ilias.d@ed.ac.uk}. Supported in part by a SICSA PECE grant and a Carnegie research grant.} \\
University of Edinburgh\\
\and Rocco A.\ Servedio\thanks{{\tt rocco@cs.columbia.edu}. Supported by NSF grant CCF-1115703.}\\
Columbia University\\
}

\begin{document}

\setcounter{page}{0}

\maketitle

\thispagestyle{empty}

\begin{abstract}

Let $g: \{-1,1\}^k \to \{-1,1\}$ be any Boolean function
and $q_1,\dots,q_k$ be any degree-2 polynomials over
$\{-1,1\}^n.$
We give a \emph{deterministic} algorithm which,
given as input explicit descriptions of
$g,q_1,\dots,q_k$ and an accuracy parameter $\eps>0$,
approximates
\[
\Pr_{x \sim \{-1,1\}^n}[g(\sign(q_1(x)),\dots,\sign(q_k(x)))=1]
\]
to within an additive $\pm \eps$.  For any constant $\eps > 0$
and $k \geq 1$ the running time of our algorithm is a fixed
polynomial in $n$ (in fact this is true even for some not-too-small
$\eps = o_n(1)$ and not-too-large $k = \omega_n(1)$).
This is the first fixed polynomial-time algorithm
that can deterministically approximately count
satisfying assignments of a natural
class of depth-3 Boolean circuits.

Our algorithm extends a recent result \cite{DDS13:deg2count}
which gave a deterministic
approximate counting algorithm for a single degree-2 polynomial
threshold function $\sign(q(x)),$ corresponding to the $k=1$ case of our
result.  Note that even in the $k=1$ case it is NP-hard to determine
whether $\Pr_{x \sim \{-1,1\}^n}[\sign(q(x))=1]$ is nonzero,
so any sort of multiplicative approximation is almost certainly
impossible even for efficient randomized algorithms.

Our algorithm and analysis requires several novel technical ingredients
that go significantly beyond the tools required to handle the $k=1$ case
in \cite{DDS13:deg2count}.  One of these
is a new multidimensional central limit theorem 
for degree-2 polynomials in Gaussian random variables which builds
on recent Malliavin-calculus-based results from probability theory.  We
use this CLT as the basis of a new decomposition technique for $k$-tuples
of degree-2 Gaussian polynomials and thus obtain an efficient
deterministic approximate counting 
algorithm for the Gaussian distribution, i.e., an algorithm for estimating
\[
\Pr_{x \sim N(0,1)^n}[g(\sign(q_1(x)),\dots,\sign(q_k(x)))=1].
\]
Finally, a third new ingredient
is a ``regularity lemma'' for \emph{$k$-tuples} of degree-$d$ 
polynomial threshold functions.  This generalizes both the regularity lemmas
of \cite{DSTW:10,HKM:09} 
(which apply to a single degree-$d$ polynomial threshold
function) and the regularity lemma of Gopalan et al \cite{GOWZ10} 
(which applies to
a $k$-tuples of \emph{linear} threshold functions, i.e., the case $d=1$).
Our new regularity lemma lets us extend our deterministic approximate
counting results from the Gaussian to the Boolean domain.
\ignore{but their approach does not seem to extend to larger values of $d$.
Via a different approach we give a regularity lemma that holds for all
values of $k$ and $d$.}

\end{abstract}

\newpage

\newpage

\section{Introduction}

\ignore{

}

Unconditional derandomization has been an important research area in computational complexity theory
over the past two decades~\cite{AjtaiWigderson:85, Nis91, Nisan:92, NiW94}. A major research goal in this area
is to obtain efficient deterministic approximate counting algorithms for ``low-level'' complexity classes
such as constant depth circuits, small space branching programs, polynomial threshold functions, and others
\cite{LVW93, LubyVelickovic:96, Trevisan:04, GopalanMR:13,
Viola09, GKMSVV11, DDS13:deg2count}.
Under the widely-believed hypothesis $\mathbf{P} = \mathbf{BPP}$, there must be a polynomial time deterministic algorithm
that can approximate the fraction of satisfying assignments to any polynomial--size circuit.
Since finding such an algorithm  seems to be out of reach of present day complexity theory~\cite{kabimp02}, research efforts have been directed
to the aforementioned low-level classes.

A natural class of Boolean functions to consider in this context is the class of polynomial threshold functions (PTFs).
Recall that a
degree-$d$ PTF, $d \ge 1$, is a Boolean function $f: \{-1,1\}^n \to \{-1,1\}$ defined by
$f(x) = \sign(p(x))$ where $p: \{-1,1\}^n \to \R$ is a degree-$d$
polynomial over the reals and $\sign: \R \to \{-1,1\}$ is defined
as $\sign(z) = 1$ iff $z \geq 0$.  In the special case where $d=1$, degree-$d$ PTFs are often
referred to as \emph{linear threshold functions} (LTFs).
Understanding the structure of these functions has been a topic of extensive investigation for decades
(see e.g., \cite{ MyhillKautz:61, MTT:61, MinskyPapert:68, Muroga:71, GHR:92, Orponen:92,
Hastad:94,Podolskii:09} and many other works) due to their importance in fields such as concrete complexity
theory \cite{Sherstov:08,Sherstov:09,DHK+:10,Kane:10,Kane12,Kane12GL,KRS12}, learning theory
\cite{KKMS:08,SSS11,DOSW:11,DDFS:12stoc},
voting theory \cite{APL:07, DDS12icalp}, and others.

In the context of approximate counting, there is a significant gap in our understanding of  low-degree PTFs.
An outstanding open problem is to design a deterministic algorithm that approximates the fraction of satisfying assignments
to a constant degree PTF over $\bn$ to an additive $\pm \eps$ and runs in time $\poly(n/\eps)$.
Even for the class of degree-$2$ PTFs, until recently
no deterministic algorithm was known with running time $\poly(n)$ for any sub-constant value of the error $\eps$.
In previous work~\cite{DDS13:deg2count} we obtained such an algorithm.
In the present paper we make further progress on this problem by developing the first efficient deterministic counting algorithm
for the class of {\em juntas of (any constant number of) degree-$2$ PTFs}.

\subsection{ Our main result.}

As our main result, we give a polynomial-time deterministic approximate counting algorithm for
any Boolean function of constantly many degree-2 polynomial threshold functions.

\begin{theorem} \label{thm:main-boolean-informal}
[Deterministic approximate counting of functions of degree-2 PTFs
over $\{-1,1\}^n$]
There is an algorithm with the following properties:  given
an arbitrary function $g: \{-1,1\}^k \to \{-1,1\}$ and $k$
degree-2 polynomials $q_1(x_1,\dots,x_n),\dots,q_k(x_1,\dots,x_n)$
and an accuracy parameter $\eps>0$, the
algorithm runs (deterministically) in
time $\poly(n) \cdot 2^{(1/\eps)^{2^{O(k)}}} $
and outputs a value $v \in [0,1]$ such that
\[
\left|\Pr_{x \in \{-1,1\}^n}[g( \sgn(q_1(x)),\dots, \sgn(q_k(x)))=1] - v \right| \leq \eps.
\]
\end{theorem}

Our result may be (somewhat informally) restated in terms of Boolean circuits as a
$\poly(n)$-time deterministic approximate counting algorithm for the class
$\NC^0$-$\Thr$-$\AND_2$ of depth-$3$ circuits that have an arbitrary $\NC^0$ gate (i.e., junta) at the top level,
arbitrary weighted threshold gates at the middle level, and fanin-2 $\AND$ gates at the bottom level.
Theorem \ref{thm:main-boolean-informal} is a broad generalization 
of the main result of
\cite{DDS13:deg2count}, which establishes the special
$k=1$ case of the current result. 

As noted in \cite{DDS13:deg2count}, the problem of
determining whether $\Pr_{x \in \{-1,1\}^n}[p(x) \geq 0]$ is nonzero
for a degree-$2$ polynomial $p$  is well known to be NP-hard, and hence
no efficient algorithm, even allowing randomness, can give a
multiplicative approximation to $\Pr_{x \sim \{-1,1\}^n}[p(x) \geq 0]$
unless NP $\subseteq$ RP.  Given this, it is natural to work towards an
additive approximation, which is what we achieve.

\smallskip

\noindent {\bf Previous work.}
For $k=1$ and $d=1$ Gopalan \etal in \cite{GKMSVV11}
obtained a multiplicatively $(1 \pm \eps)$-accurate deterministic
$\poly(n,1/\eps)$ time approximate counting algorithm.
For $d \geq 2$, however, as noted above additive approximation is the best one
can hope for. For the special case of $k=1$, in separate work
\cite{DDS13:deg2count}, the authors have given a deterministic approximate
counting algorithm that runs in time $\poly(n,2^{\poly(1/\eps)})$.
As we explain in detail in the rest of this introduction,
more sophisticated ideas and techniques
are required to obtain the results of the current paper for general $k$.
These include a new central limit theorem based on Malliavin calculus and Stein's method,
and a new decomposition procedure that goes well beyond the decomposition
approach employed in \cite{DDS13:deg2count}.

We remark that the only previous deterministic approximate counting
algorithm for $k$-juntas of degree-$2$ PTFs follows
from the  \emph{pseudorandom generators} (PRGs) of~\cite{DKNfocs10} (which are based on bounded independence).
The running time of the resulting algorithm is $n^{\poly(1/\eps)}$, even for $k=1$.

\subsection{Techniques.}

Our high-level approach to establishing Theorem \ref{thm:main-boolean-informal}
follows a by now standard approach in this area.  We first (i)
establish the result for general polynomials over Gaussian inputs; then
(ii)
use a ``regularity lemma'' to show that
every polynomial over Boolean inputs can be decomposed into a ``small'' number
of regular polynomials over Boolean inputs; and finally
(iii)
use an invariance principle to reduce the problem for ``regular''
polynomials over Boolean inputs to the problem for regular polynomials over
Gaussian inputs.
This general approach has been used in a number
of previous works, including constructions of unconditional PRGs~\cite{DGJ+:10, MZstoc10, GOWZ10, DKNfocs10, Kane11ccc, Kane12},
learning and property testing~\cite{MORS:10, OS11:chow}, and other works.
However, we emphasize that significant novel conceptual and technical
work is required to make this approach work in our setting.
More specifically, to achieve step (i), we require (i.a) a new multidimensional CLT
for degree-$2$ Gaussian polynomials with small eigenvalues and (i.b) a new decomposition
procedure that transforms a $k$-dimensional vector of Gaussian polynomials into a tractable
form for the purpose of approximate counting. For step (ii) we establish a novel regularity lemma
for $k$-vectors of low-degree polynomials. Finally, Step (iii) follows by an application of the invariance
principle of Mossel~\cite{Mossel10} combined with appropriate mollification arguments~\cite{DKNfocs10}.
In the rest of this section we discuss
our new approaches to Steps (i) and (ii).

\vspace{-0.3cm}

\paragraph{Step (i):  The counting problem over Gaussian inputs.}
The current paper goes significantly beyond the techniques of
\cite{DDS13:deg2count}.  To explain our new contributions let us first
briefly recall the \cite{DDS13:deg2count} approach.

The main observation enabling the result in~\cite{DDS13:deg2count} is this: Because of rotational symmetry of the Gaussian distribution, a degree-$2$
Gaussian polynomial can be ``diagonalized'' so that there exist no ``cross-terms'' in its representation. In a little more detail,
if $p(x) = \littlesum_{i, j} a_{ij}x_{i}x_{j}$ (we ignore the linear term for simplicity), where $x \sim N(0,1)^n$, then $p$ can be rewritten
in the form $p(y) = \littlesum_{i} \lambda_i y_i^2$, where 
$y \sim N(0,1)^n$ and the $\lambda_i$'s are the eigenvalues of the corresponding matrix.
Roughly speaking, once such a representation has been (approximately) constructed, the counting problem can be solved efficiently by dynamic programming.
To construct such a decomposition, \cite{DDS13:deg2count} employs a ``critical-index'' based analysis on the eigenvalues of the corresponding matrix.
For the analysis of  the \cite{DDS13:deg2count} algorithm, \cite{DDS13:deg2count} proves a CLT for a single
degree-$2$ Gaussian polynomial with small eigenvalues (this CLT is based on a result
of Chaterjee~\cite{Chatterjee:09}). (We note that this informal description suppresses several  non-trivial technical issues,
see~\cite{DDS13:deg2count} for details.)

At a high level, the approach of the current paper builds on the approach of \cite{DDS13:deg2count}. To solve the Gaussian counting problem
we use a combination of (i.a) a new multidimensional CLT for $k$-tuples of degree-$2$ Gaussian polynomials with small eigenvalues, and (i.b)
a novel decomposition result for $k$-tuples of degree-2 Gaussian polynomials. We now elaborate on these steps.
\begin{itemize}

\item[(i.a)]  As our first contribution, we prove a new multidimensional central limit theorem
for $k$-tuples of degree-$2$ Gaussian polynomials (Theorem~\ref{thm:small-dk}).
Roughly speaking, our CLT states that if
each polynomial in the $k$-tuple has small eigenvalues, then the joint distribution of the $k$-tuple
is close to a $k$-dimensional Gaussian random variable with matching mean and covariance matrix.
The closeness here is with respect to the $k$-dimensional Kolmogorov distance over $\R^k$ (a
natural generalization of Kolmogorov distance to vector-valued random variables, which we denote
$\dk$ and which is useful for
analyzing  PTFs).
To establish our new CLT, we proceed in two steps: In the first (main) step,
we make essential use of a recent multidimensional CLT due to Nourdin and Peccati~\cite{NourdinPeccati09} (\ifnum\confversion=0
Theorem~\ref{thm:Nourdin}\else Theorem 11 in the full version\fi)
which is proved using a combination of Malliavin calculus and Stein's method.
To use this theorem in our setting, we perform a linear-algebraic analysis which allows us to obtain precise bounds on the
Malliavin derivatives of degree-$2$ Gaussian polynomials with small eigenvalues.
An application of~\cite{NourdinPeccati09} then gives us a version of our desired CLT with respect
to ``test functions'' with bounded second derivatives (\ifnum\confversion=0Theorem~\ref{thm:h-match-cov}\else Theorem 12 in the full version\fi).
In the second step, we use tools from mollification~\cite{DKNfocs10} to translate this notion of closeness
into  closeness with respect to $k$-dimensional Kolmogorov distance,  thus obtaining our intended CLT.
(As a side note, we believe that this work is the first to use Malliavin-calculus-based tools in the context of derandomization.)

\item[(i.b)]  As our second contribution, we give an efficient procedure that transforms a $k$-tuple
of degree-$2$ Gaussian polynomials $p = (p_1, \ldots, p_k)$ into a $k$-tuple of  degree-$2$ Gaussian polynomials
$r = (r_1, \ldots, r_k)$ such that: (1) $p$ and $r$ are $\dk$-close, and (2) the $k$-tuple
$r$ has a ``nice structure'' that allows for efficient deterministic approximate counting. In particular, there is a ``small'' set
of variables such that for each restriction $\rho$ fixing  this set, the restricted $k$-tuple of polynomials $r|_{\rho}$ is well-approximated
by a $k$-dimensional Gaussian random variable (with the appropriate mean and covariance matrix). Once such an $r$ has been obtained,
deterministic approximate counting is straightforward via an appropriate discretization of the $k$-dimensional
Gaussian distribution (see Section~\ref{sec:gauss}).

We now elaborate on Item (1) above. At a high level, the main step of our transformation procedure performs a ``change of basis''
to convert $p=(p_1(x),\dots,p_k(x))$ into an essentially equivalent (for the purpose of approximate counting)
vector $q=(q_1(y),\dots,q_k(y))$ of polynomials.
The high-level approach to achieve this is reminiscent of (and inspired by) the
decomposition procedure for vectors of $k$ linear forms in \cite{GOWZ10}. However, there are significant
complications that arise in our setting.
In particular, in the \cite{GOWZ10} approach, a vector of $k$ linear forms
is simplified by ``collecting'' variables in a greedy fashion as follows:  Each of the
$k$ linear forms has a ``budget'' of at most $B$, meaning that at most
$B$ variables will be collected on its behalf. Thus, the overall number
of variables that are collected is at most $kB$.  At each stage some variable is collected which has large influence in the
remaining (uncollected) portion of some linear form.  The \cite{GOWZ10}
analysis shows that after at most $B$ variables have been collected on behalf of each linear form,
each of the $k$ linear forms will either be regular or its remaining portion
(consisting of the uncollected variables) will have small variance.
In our current setting, we are dealing with $k$ degree-$2$ Gaussian polynomials
instead of $k$ linear forms. Recall that
every degree-$2$ polynomial can be expressed as a linear combination
of squares of linear forms (i.e., it can be diagonalized). Intuitively, since Gaussians are invariant under change of basis,
we can attempt to use an approach where linear forms will play the role that variables had in \cite{GOWZ10}.
Mimicking the \cite{GOWZ10} strategy, each quadratic polynomial will have at most $B$ linear forms
collected on its behalf, and at most $kB$ linear forms will be collected
overall. Unfortunately, this vanilla strategy does not work even for $k=2$, as it requires a single
orthonormal basis in which all the degree-$2$ polynomials are simultaneously diagonalized.

Instead, we resort to a more refined strategy. Starting with the $k$ quadratic polynomials, we use the following iterative algorithm:
If the largest magnitude eigenvalue of each quadratic form is small, we are already in the \emph{regular}
case (and we can appeal to our multidimensional CLT).
Otherwise, there exists at least one polynomial with a large magnitude eigenvalue. We proceed to collect the corresponding linear form
and ``reduce" every polynomial by this linear form. (The exact description of this reduction is somewhat involved to describe, but intuitively,
it uses the fact that Gaussians are invariant under orthogonal transformations.) This step is repeated iteratively; an argument similar to \cite{GOWZ10}
shows that for every quadratic polynomial, we can collect at most $B$ linear forms. At the end of this procedure,  each of the $k$ quadratic polynomials will either be
``regular'' (have small largest magnitude eigenvalue compared to the variance of the remaining portion), or else the variance of the remaining portion will
be small. This completes the informal description of our transformation.
\end{itemize}

Our main result for the Gaussian setting is the following theorem:

\begin{theorem} \label{thm:maingauss-informal}
[Deterministic approximate counting of functions of
degree-$2$ PTFs over Gaussians]
There is an algorithm with the following properties:
It takes as input explicit descriptions of $n$-variable degree-$2$
polynomials $q_1,\dots,q_k$,  an
explicit description of a $k$-bit Boolean function
$g: \{-1,1\}^k \to \{-1,1\}$, and a value $\eps > 0.$
It runs (deterministically) in time $\poly(n) \cdot 2^{\poly(2^k/\eps)}$
and outputs a value $\tilde{v} \in [0,1]$ such that
\begin{equation} \label{eq:close2}
\left|\Pr_{{\cal G} \sim N(0,1)^n}[g(Q_1({\cal G}),\dots,Q_k({\cal G}))=1] - \tilde{v} \right|
\leq \eps,
\end{equation}
where $Q_i(x) = \sign(q_i(x))$ for $i=1,\dots,k.$
\end{theorem}

We note that in the case $k=1$, the algorithm of the current work is not the same as the algorithm of
\cite{DDS13:deg2count} (indeed, observe the above algorithm runs in time exponential in $1/\eps$ even for $k=1$, whereas
the algorithm of~\cite{DDS13:deg2count} runs in time $\poly(n/\eps)$ for a single Gaussian polynomial).

\paragraph{Step (ii):  The regularity lemma.}
Recall that the \emph{influence} of variable $i$ on a multilinear polynomial $p=\sum_{S \subseteq [n]}
\widehat{p}(S) \prod_{i \in S} x_i$ over $\{-1,1\}^n$ (under the uniform distribution) is $\Inf_i(p) \eqdef \sum_{S \ni i} \widehat{p}(S)^2$
and that the \emph{variance} of $p$ is
$\Var[p] = \E_{x \in \{-1,1\}^n}[(p(x) - \E[p])^2] =
\sum_{\emptyset \neq S} \widehat{p}^2(S)$.
For $p$ a degree-$d$ polynomial we have $\Var[p] \leq \sum_{i=1}^n \Inf_i(p)
\leq d \cdot \Var[p],$ so for small constant $d$ the variance and the total
influence $\sum_{i=1}^n \Inf_i(d)$ are equal up to a small constant factor.
A polynomial $p$ is said to be \emph{$\tau$-regular} if for all $i \in [n]$
we have $\Inf_i(p) \leq \tau \cdot \Var[p].$

As noted earlier, by adapting known invariance principles from the literature \cite{Mos08} it is possible
to show that an algorithm for approximately counting satisfying assignments of a junta of degree-2 PTFs
over $N(0,1)^n$ will in fact also succeed for approximately counting satisfying assignments of a junta
of sufficiently regular degree-2 PTFs over $\{-1,1\}^n$.  Since
Theorem \ref{thm:maingauss-informal}
gives us an algorithm for the Gaussian problem, to complete the chain we need a reduction from
the problem of counting satisfying assignments of a junta of \emph{arbitrary} degree-2 PTFs over $\{-1,1\}^n$,
to the problem of counting satisfying assignments of a junta of \emph{regular} degree-2 PTFs over $\{-1,1\}^n$.

We accomplish this by giving a novel \emph{regularity lemma} for $k$-tuples of degree-2 (or more
generally, degree-$d$) polynomials.
Informally speaking, this is an efficient deterministic algorithm with the following
property:  given as input a $k$-tuple of arbitrary degree-2 polynomials $(p_1,\dots,p_k)$
over $\{-1,1\}^n$, it constructs
a decision tree of restrictions such that for almost every root-to-leaf path (i.e., restriction $\rho$)
in the decision tree, \emph{all $k$} restricted polynomials $(p_1)_\rho,\dots,(p_k)_\rho$ are ``easy to handle''
for deterministic approximate counting, in the following sense:  each $(p_i)_\rho$ is either highly regular,
or else is highly \emph{skewed}, in the sense that its constant term is so large compared to its variance
that the corresponding PTF $\sign((p_i)_\rho)$ is guaranteed to be very close to a
constant function.  Such leaves are ``easy to handle'' because we can set the PTFs corresponding
to ``skewed'' polynomials to constants (and incur only small error); then we are left
with a junta of regular degree-2 PTFs, which can be handled using the Gaussian algorithm as sketched above.

A range of related ``regularity lemmas'' have been given in the LTF/PTF literature
\cite{DSTW:10,HKM:09,BLY:09,GOWZ10}, but none with all the properties that we require.
\cite{Servedio:07cc} implicitly gave a regularity lemma for a single LTF, and
\cite{DSTW:10,HKM:09,BLY:09} each gave (slightly different flavors of) regularity lemmas for a single
degree-$d$ PTF.  Subsequently \cite{GOWZ10} gave a regularity lemma for $k$-tuples of LTFs;
as noted earlier our decomposition for $k$-tuples of degree-2 polynomials over Gaussian inputs given
in Section \ref{sec:gauss} uses ideas
from their work.  However, as we describe in Section \ifnum\confversion=0\ref{sec:regularity}\else7 of the full version\fi, their approach does not seem
to extend to degrees $d>1$, so we must use a different approach to prove our regularity
lemma.

\subsection{Organization.}

\ifnum\confversion=0
After giving some useful background in Section \ref{sec:prelim},
we prove our new multidimensional CLT in Section \ref{sec:structure}.
We give the transformation procedure that is at the heart of our
decomposition approach in Section \ref{sec:gauss-decomp}, and
present the actual deterministic counting algorithm for the Gaussian
case that uses this transformation in Section \ref{sec:gauss}.
Section \ref{sec:Bool} shows how the new regularity lemma for $k$-tuples of
Boolean PTFs gives the main Boolean counting result, and
finally the regularity lemma is proved in Section \ref{sec:regularity}.
\else
In Section \ref{sec:structure} we present the statement of our new multidimensional CLT and sketch the main ingredients required for its proof.
We give the transformation procedure that is at the heart of our
decomposition approach in Section \ref{sec:gauss-decomp}, and
present the actual deterministic counting algorithm for the Gaussian
case that uses this transformation in Section \ref{sec:gauss}.
Section \ref{sec:Bool} shows how the new regularity lemma for $k$-tuples of
Boolean PTFs gives the main Boolean counting result. Due to space limitations, the regularity lemma is proved in the full version.
\fi

\section{Definitions, Notation and Useful Background} \label{sec:prelim}

\paragraph{Polynomials and PTFs.}
Throughout the paper we use lower-case letters $p,q,$ etc. to denote
low-degree multivariate polynomials.  We use capital letters
to denote the corresponding polynomial threshold functions that
map to $\{-1,1\}$, so typically
$P(x) = \sign(p(x))$, $Q(x) = \sign(q(x))$, etc.

We consider multivariate polynomials over the domains $\R^n$ (endowed with the
standard normal distribution $N(0,1)^n$) and $\{-1,1\}^n$ (endowed with the
uniform distribution).  Since $x^2=1$ for $x \in \{-1,1\}$, in dealing with
polynomials over the domain $\{-1,1\}^n$ we may without loss of generality
restrict our attention to multilinear polynomials.

\paragraph{Kolmogorov distance between $\R^k$-valued random variables.}
It will be convenient for us to
use a natural $k$-dimensional generalization of the Kolmogorov
distance between two real-valued random variables which we now describe.
Let $X=(X_1,\dots,X_k)$
and $Y=(Y_1,\dots,Y_k)$ be two $\R^k$-valued random variables.  We define the
\emph{$k$-dimensional Kolmogorov distance} between $X$ and $Y$ to be
\[
\dk(X,Y) = \sup_{(\theta_1,\dots,\theta_{k}) \in \R^{k}}
\left| \Pr [\forall \ i \in [{k}] \ X_i \leq \theta_i]
 - \Pr [\forall \ i \in [{k}] \ Y_i \leq \theta_i]  \right|.
\]
This will be useful to us when we are analyzing $k$-juntas
of degree-2 PTFs over Gaussian random variables; we will typically
have $X=(q_1(x),\dots,q_k(x))$ where $x \sim N(0,1)^n$
and $q_i$ is a degree-2 polynomial, and have $Y=(Y_1,\dots,Y_k)$
be a $k$-dimensional Gaussian random variable whose mean and covariance
matrix match those of $X$.

\paragraph{Notation and terminology for degree-2 polynomials.}
Let $q=(q_1(x),\dots,q_k(x))$ be a vector of polynomials over $\R^n$.  We
endow $\R^n$ with the $N(0,1)^n$ distribution, and hence we may view
$q$ as a $k$-dimensional random variable. We sometimes refer
to the $q_i$'s as \emph{Gaussian polynomials}.

For $A$ a real $n \times n$
matrix we write $\|A\|_2$ to denote the operator
norm
$
\|A\|_2 = \max_{0 \neq x \in \R^n} {\frac {\|Ax\|_2}{\|x\|_2}}.$

Given a degree-$2$ polynomial $q : \mathbb{R}^n \rightarrow \mathbb{R}$
defined as $q(x) = \littlesum_{1 \le i \le j \le n} a_{ij} x_i x_j
+ \littlesum_{1 \le i \le n} b_i x_i + C$, we define the (symmetric)
matrix $A$ corresponding to its quadratic part as : $A_{ij} = a_{ij}
(1/2 + \delta_{ij}/2)$.
Note that with this definition we have that
$x^T \cdot A \cdot x = \littlesum_{1 \le i \le j \le n} a_{ij} x_i x_j$
for the vector $x = (x_1, \ldots, x_n)$.

Throughout the paper we adopt the convention that the eigenvalues
$\lambda_1,\dots,\lambda_n$ of a real symmetric matrix $A$
satisfy $|\lambda_1| \geq \cdots \geq |\lambda_n|$.
We sometimes write $\lambda_{\max}(A)$ to denote $\lambda_1$,
and we sometimes write $\lambda_i(q)$ to refer to the $i$-th eigenvalue
of the matrix $A$ defined based on $q$ as described above.

\paragraph{Degree-2 polynomials and their heads and tails.}
The following notation will be useful for us,
especially in Section \ref{sec:gauss-decomp}.
Let $z(y_1,\dots,y_n) = \sum_{1 \leq i \leq j \leq n}
a_{ij} y_iy_j + \sum_{1 \leq i \leq n} b_i y_i + c$
be a degree-2 polynomial.  For $0 \leq t \leq n$ we say the
\emph{$t$-head of $z(y)$}, denoted $\head_t(z(y))$, is the polynomial

\begin{equation} \label{eq:head}
\head_t(z(y)) \eqdef \sum_{1 \leq i \leq t, j \geq i}
a_{ij} y_iy_j + \sum_{1 \leq i \leq t} b_i y_i
\end{equation}
and the
\emph{$t$-tail of $z(y)$}, denoted $\tail_t(z(y))$, is the polynomial

\begin{equation} \label{eq:tail}
\tail_t(z(y)) \eqdef \sum_{t < i \leq j \leq n}
a_{ij} y_iy_j + \sum_{t < i \leq n} b_i y_i + c,
\end{equation}
so clearly we have $z(y)=\head_t(z(y))+\tail_t(z(y))$.
(Intuitively, $\tail_t(z(y))$ is the part of $z(y)$ which does not
``touch'' any of the first $t$ variables $y_1,\dots,y_t$ and
$\head_t(z(y))$ is the part which does ``touch'' those variables.)

\begin{remark} \label{rem:restric}
Note that if $\rho=(\rho_1,\dots,\rho_t) \in \R^t$
is a restriction fixing variables $y_1,\dots,y_t$,
then the restricted polynomial $z|_\rho(y) \eqdef z(\rho_1,\dots,\rho_t,
y_{t+1},\dots,y_n)$ is of the form $\tail_t(z(y)) + L(y_{t+1},\dots,y_n)$
where $L$ is an affine form.
\end{remark}

We further define $\quadtail_t(z(y))$, the ``quadratic portion of the
$t$-tail,'' to be
\begin{equation} \label{eq:quadtail}
\quadtail_t(z(y)) \eqdef \sum_{t < i \leq j \leq n} a_{ij} y_iy_j.
\end{equation}

Setting aside heads and tails, it will sometimes be useful for us to consider the sum of squares
of all the (non-constant) coefficients of a degree-2 polynomial.  Towards that end we have the following
definition:
\begin{definition}
Given $p : \mathbb{R}^n \rightarrow \mathbb{R}$ defined by
$p(x) =
\littlesum_{1\le i \le j \le n} a_{ij} x_i x_j + \littlesum_{1 \le i \le n}
b_i x_i + C$,
define $\mathop{SS}(p)$ as $\mathop{SS}(p)=
\littlesum_{1\le i \le j \le n} a_{ij}^2 + \littlesum_{1 \le i \le n} b_i^2$.
\end{definition}

The following straightforward claim is established in
\cite{DDS13:deg2count}:
\begin{claim}\label{clm:variance-bound-SS} [Claim 20 of \cite{DDS13:deg2count}]
Given $p : \mathbb{R}^n \rightarrow \mathbb{R}$,
we have that $2 \mathop{SS}(p) \ge \Var(p) \ge  SS(p)$.
\end{claim}

\paragraph{Tail bounds and anti-concentration bounds
on low-degree polynomials in Gaussian variables.}

We will need the following standard concentration bound for
low-degree polynomials over independent  Gaussians.

\begin{theorem}[``degree-$d$ Chernoff bound'',  \cite{Janson:97}] \label{thm:dcb}
Let $p: \R^n \to \R$ be a degree-$d$ polynomial. For any
$t > e^d$, we have
\[
\Pr_{x \sim N(0,1)^n}[
|p(x) - \E[p(x)]| > t  \cdot \sqrt{\Var(p(x))} ] \leq
{d e^{-\Omega(t^{2/d})}}.
\]
\end{theorem}

\noindent We will also use the following anti-concentration bound for
degree-$d$ polynomials over Gaussians:

\begin{theorem}[\cite{CW:01}]
\label{thm:cw} Let $p: \R^n \to \R$ be a degree-$d$ polynomial
that is not identically 0.  Then for all $\eps>0$ and all
$\theta \in \R$, we have
\[
\Pr_{x \sim N(0,1)^n}\left[|p(x) - \theta| < \eps \sqrt{\Var(p)}
\right] \le O(d\eps^{1/d}).
\]
\end{theorem}

\medskip

\noindent {\bf The model.}
Throughout this paper, our algorithms will repeatedly be performing basic linear algebraic operations, in particular SVD computation
and Gram-Schmidt orthogonalization.  In the bit complexity model, it is well-known that these linear algebraic operations can be performed
 (by deterministic algorithms) up to additive error $\epsilon$ in time $\poly(n, 1/\epsilon)$. For example, let $A \in \mathbb{R}^{n \times m}$ have $b$-bit rational
 entries.  It is known (see \cite{Golub} for details) that in time $\poly(n,m,b,1/\epsilon)$, it is possible to compute a value $\tilde{\sigma}_1$ and
 vectors $u_1 \in \R^n$, $v_1 \in \R^m$, such that $\tilde{\sigma}_1 = {\frac {u_1^T A v_1} {\|u_1\| \|v_1\|}}$ and $|\tilde{\sigma}_1 - \sigma_1| \leq \epsilon$,
 where $\sigma_1$ is the largest singular value of $A$. Likewise, given $n$ linearly independent vectors $v^{(1)},
 \dots,v^{(n)} \in \R^m$ with $b$-bit rational entries, it is possible
 to compute vectors $\tilde{u}^{(1)},\dots,\tilde{u}^{(n)}$ in
 time $\poly(n,m,b)$ such that
 if $u^{(1)},\dots,u^{(n)}$ is a Gram-Schmidt orthogonalization of
 $v^{(1)},\dots,v^{(n)}$ then we have $|u^{(i)} \cdot u^{(j)} -
 \tilde{u}^{(i)} \cdot
 \tilde{u}^{(j)}| \leq 2^{-\poly(b)}$ for all $i,j$.

 In this paper, we work in a unit-cost real number model of computation.
This allows us to assume that given a real matrix $A \in \mathbb{R}^{n \times m}$  with $b$-bit rational entries, we can compute the SVD of $A$ exactly in time
$\poly(n,m,b)$. Likewise, given $n$ vectors over $\mathbb{R}^m$, each of whose entries are $b$-bit rational numbers, we can perform an exact Gram-Schmidt orthogonalization
in time $\poly(n,m,b)$. Using high-accuracy approximations of the sort 
described above throughout our algorithms,
it is straightforward to translate our unit-cost real-number
algorithms into the bit complexity setting, at the cost of some additional
 error in the resulting bound.

 Using these two observations, it can be shown that by making
 sufficiently accurate approximations at each stage where a numerical
computation is performed by our ``idealized'' algorithm,
 the cumulative error resulting from all of the approximations
 can be absorbed into the final $O(\eps)$ error bound.  Since inverse
 polynomial levels of error can be achieved in polynomial time
 for all of the approximate numerical computations that our algorithm
performs,
and since only poly$(n)$ many such approximation steps are performed by
 poly$(n)$-time algorithms, the resulting approximate implementations
 of our algorithms in a bit-complexity model
 also achieve the guarantee of our main results,
 at the cost of a  fixed $\poly(n)$ overhead in the running time.
For the sake of completeness, such a detailed numerical analysis was performed in our previous paper~\cite{DDS13:deg2count}.
 Since working through the details of such an analysis is tedious and detracts from the clarity of the presentation, 
 we content ourselves with this brief discussion in this work.

\section{A multidimensional CLT for degree-2 Gaussian polynomials}
\label{sec:structure}

In this section we \ifnum\confversion=0prove \else present \fi a central limit theorem
which plays a crucial role in the decomposition
result which we establish in the following sections.
Let $q=(q_1,\dots,q_k)$ where each $q_i$ is a
degree-2 polynomial in Gaussian random variables $(x_1,\dots,x_n)
\sim N(0,1)^n.$ Our CLT states
that under suitable conditions on $q_1,\dots,q_k$ ---
all of them have only small--magnitude eigenvalues,
no $\Var[q_i]$ is too large and at least one $\Var[q_i]$ is not
too small --- the distribution of
$q$ is close (in $k$-dimensional Kolmogorov distance)
to the distribution of the $k$-dimensional Gaussian random variable
whose mean and covariance matrix match $q$. 

\ifnum\confversion=1
Let $X=(X_1,\dots,X_k)$ and $Y=(Y_1,\dots,Y_k)$ be two $\R^k$-valued random variables.  We define the
\emph{$k$-dimensional Kolmogorov distance} between $X$ and $Y$ to be
$ \dk(X,Y) = \sup_{(\theta_1,\dots,\theta_{k}) \in \R^{k}} |\Pr [\forall \ i \in [{k}] \ X_i \leq \theta_i] - \Pr [\forall \ i \in [{k}] \ Y_i \leq \theta_i]|.
$ We have the following:
\fi

\begin{theorem}\label{thm:small-dk}
Let $q=(q_1(x),\dots,q_k(x))$ where each $q_i$ is a
degree-2 Gaussian polynomial
that satisfies
 $\Var[q_i] \leq 1$ and $|\lambda_{\max}(q_i)| \le \eps$ for all $i \in [k]$.
Suppose that $\max_{i \in [k]} \Var(q_i) \ge \lambda$.
Let $C$ denote the covariance matrix of $q$ and let $
N=N((\mu_1,\dots,\mu_k),C)$ be a
$k$-dimensional Gaussian random variable with covariance matrix
$C$ and mean $(\mu_1, \ldots, \mu_k)$ where $\mu_i = \mathbf{E}[q_i]$.  Then
\ifnum\confversion=1
$\dk(q,N) \leq O (k^{2/3} \eps^{1/6}/\lambda^{1/6}).$
\else
\[
\dk(q,N) \leq O\left({\frac {k^{2/3} \eps^{1/6}}{\lambda^{1/6}}}\right).
\] \fi
\end{theorem}
Looking ahead to motivate this result for our ultimate purposes,
Theorem \ref{thm:small-dk} is useful for
deterministic approximate counting because if $q=(q_1,\dots,q_k)$
satisfies the conditions of the theorem, then the theorem ensures that
$\Pr_{x \sim N(0,1)^n} \left[ \forall \ell \in [k], q_\ell (x) \leq 0 \right]$
is close to $\Pr\left[ \forall \ell \in [k],   N_\ell \leq 0 \right]$.  Note 
that the latter quantity can be efficiently estimated by a deterministic algorithm.

A key ingredient in the proof of
Theorem \ref{thm:small-dk} is a CLT due to Nourdin and Peccati
\cite{NourdinPeccati09} which gives a bound that involves
the Malliavin derivative of the functions $q_1,\dots,q_k$.
In Section \ifnum\confversion=0\ref{sec:Malliavin}\else 3.1 of the full version \fi  we give the necessary background
from Malliavin calculus and build on the \cite{NourdinPeccati09} result
to prove a result which is similar to Theorem
\ref{thm:small-dk} but gives a bound on $\E[h(q)]-\E[h(N)]$ rather
than $\dk(q,N)$ for a broad class of ``test functions'' $h$
 \ifnum\confversion=0(see Theorem \ref{thm:h-match-cov} below)\fi.
In Section  \ifnum\confversion=0\ref{sec:mollification}\else 3.2 of the full version \fi we
show how  \ifnum\confversion=0Theorem \ref{thm:h-match-cov}\else this result \fi can be combined with
standard ``mollification'' techniques to
yield Theorem \ref{thm:small-dk}.

\ifnum\confversion=0
\subsection{Malliavin calculus and test functions
with bounded second derivative.} \label{sec:Malliavin}
We need some notation and conceptual background before we can state
the Nourdin-Peccati multi-dimensional CLT from
\cite{NourdinPeccati09}.  Their CLT is proved using Stein's method;
while there is a rich theory underlying their result
we give only the absolute basics that suffice for our purposes.  (See
e.g. \cite{NourdinPeccati09,Nourdin-notes} for detailed
treatments of Malliavin calculus and its interaction with Stein's Method.)

We will use $\mathcal{X}$ to denote the space $\R^n$
endowed with the standard $N(0,1)^n$ normal measure
and ${\cal P}$ to denote the family of all polynomials
over $\mathcal{X}$.  For integer $d \geq 0$ we let $\mathcal{H}_d$
denote the ``$d$-th Wiener chaos'' of $\mathcal{X}$, namely the
space of all homogeneous degree-$d$ Hermite polynomials over $
\mathcal{X}.$
\ignore{
 \begin{fact}\label{fac:1}
 Every $f \in \mathcal{F}$ admits a Hermite expansion (which is convergent in the right sense). We will be dealing with polynomials, so the notion of convergence is not really relevant to us.
 \end{fact}
}
We define the operator $I_d : {\cal P} \rightarrow \mathcal{H}_d$
as follows : $I_d$ maps $p \in {\cal P}$ to the degree-$d$ part
of its Hermite expansion, so if $p$ has degree $d$ then
$p=I_0(p) + \cdots + I_d(p).$

We next define the generator of the Ornstein-Uhlenbeck semigroup.
This is the operator $L$ which is defined on ${\cal P}$ via
$$
L p = \littlesum_{q=0}^{\infty} -q \cdot I_q(p).
$$
\ignore{
 Not that it is \nnew{not} of direct interest to us but it is important to note that while a function $f \in \mathcal{F}$, it may not lie in the domain of $L$ (but note that it is of consequence to us).
\rnote{You are saying it could be the case that $f \in {\cal F}$ but
$Lf$ is not defined (the sum does not converge for some input $x$), right?}
From this point onwards, I am going to stop mentioning these irrelevant
details.
}
It is easy to see that for $p \in {\cal P}$ we have the inverse operator
$$
L^{-1}p  = \littlesum_{q=1}^{\infty} \frac{-1}{q} I_q(p).
$$

Next we introduce the notion of the \emph{Malliavin derivative}.
The Malliavin derivative operator $D$ maps a real-valued random variable
(defined over ${\cal X}$ by a differentiable
real-valued function $f: \R^n \to \R$)
to an $n$-dimensional vector of random variables in the
following way:  for $f : \mathbb{R}^n \rightarrow \mathbb{R}$,
$$
Df = \left( \frac{\partial f}{\partial x_1}, \ldots,
\frac{\partial f}{\partial x_n}\right).
$$

The following key identity provides the fundamental connection
between Malliavin Calculus and Stein's method, which is used
to prove Theorem \ref{thm:Nourdin} below:
\begin{claim} [see e.g. Equation (2.22) of \cite{NourdinPeccati09}]
Let $h : \mathbb{R} \rightarrow \mathbb{R}$ be a continuous function
with a bounded first derivative. Let $p$ and $q$ be
polynomials over $\mathcal{X}$ with $ \mathbf{E}[q]=0$. Then
$\mathbf{E} [q h(p)] = \mathbf{E} [h'(p) \cdot \langle Dp \
, \ -DL^{-1} q\rangle] $.
\end{claim}
Specializing to the case $h(x)=x$, we have
 \begin{corollary}\label{corr:a}
 Let $p$ and $q$ be finite degree polynomials over $\mathcal{X}$ with $ \mathbf{E}[q]=0$. Then,
 $\mathbf{E} [q p] = \mathbf{E} [ \langle Dp \ , \ -DL^{-1} q\rangle] $.
 \end{corollary}

We now recall the following CLT due to Nourdin and Peccati:

 \begin{theorem}\label{thm:Nourdin} [\cite{NourdinPeccati09}, see also
\cite{Nourdin-notes}, Theorem~6.1]
Let $p=(p_1, \ldots, p_k)$ where each $p_i$ is a Gaussian polynomial
with $\E[p_i]=0$.
Let $C$ be a symmetric PSD matrix in $\mathbb{R}^{k \times k}$
and let $N$ be a mean-0 $k$-dimensional
Gaussian random variable with
covariance matrix $C$. Then for any $h : \mathbb{R}^{k}
\rightarrow \mathbb{R}, h \in {\cal C}^2$
such that $\Vert h'' \Vert_{\infty} < \infty$, we have
$$
\left| \E [h(p)] - \E [h(N)] \right| < {\frac 1 2}
\Vert h'' \Vert_{\infty} \cdot \left( \sum_{i=1}^k \sum_{j=1}^k
\E [|C(i,j) - Y(i,j)|]\right)
$$
where $Y(i,j) =  \langle Dp_i \ ,  -DL^{-1} p_j\rangle$.
 \end{theorem}

We now use Theorem \ref{thm:Nourdin} to prove our main result of this subsection, which is
the following CLT for multidimensional degree-2 Gaussian polynomials with small-magnitude
eigenvalues.  Our CLT says that such multidimensional random variables must in fact be close to
multidimensional Gaussian distributions, where ``closeness'' here is measured using
test functions with bounded second derivative.  (In the next subsection we extend this
result using mollification techniques to obtain Theorem \ref{thm:small-dk}, which uses multidimensional
Kolmogorov distance.)

\begin{theorem} \label{thm:h-match-cov}
Let $q=(q_1, \ldots, q_{k})$ where each $q_i$ is a
degree-$2$ mean-0 Gaussian polynomial
with $\Var[q_i] \leq 1$ and $|\lambda_{\max}(q_i)| \le \eps$.
Let $C$ denote the covariance matrix of $q$, so $C(i,j)=
\Cov(q_i,q_j)=\E[q_i q_j].$  Let $N$ be a mean-zero $k$-dimensional
Gaussian random variable with covariance matrix $C$.  Then
for any $h : \mathbb{R}^{k} \rightarrow \mathbb{R},{h \in
{\cal C}^2}$ such that $\Vert h'' \Vert_{\infty} < {\infty}$,
we have
$$\left| \E  [h(q)] -  \E [h(N)] \right| <   O(k^2 \eps) \cdot
\Vert h'' \Vert_{\infty}.$$
 \end{theorem}

\begin{proof}
As in Theorem~\ref{thm:Nourdin}, we write $Y(a,b)$ to denote
$\langle D q_a, -DL^{-1}q_b \rangle.$
For any $1 \leq a,b \leq k$, we have
\begin{equation} \label{eq:CY}
C(a,b)=\Cov(q_a,q_b) = \E[q_a q_b] = \E[Y(a,b)],
\end{equation}
where the second equality is because $q_a$ and $q_b$ have mean 0
and the third equality is by Corollary~\ref{corr:a}.
Since $C$ is a covariance matrix and every covariance matrix
is PSD, we may apply Theorem \ref{thm:Nourdin}, and we get
that
\[
\left| \E  [h(q)] -  \E [h(N)] \right| < {\frac {k^2} 2}
\Vert h'' \Vert_{\infty} \cdot \max_{1 \leq a,b \leq k}
\E[|C(a,b)-Y(a,b)|]
= {\frac {k^2} 2}
\Vert h'' \Vert_{\infty} \cdot \max_{1 \leq a,b \leq k}
\E[|Y(a,b)-\E[Y(a,b)]|],
\]
where we used (\ref{eq:CY}) for the equality.
By Jensen's inequality we have
$\E[|Y(a,b)-\E[Y(a,b)]|] \leq \sqrt{\Var[Y(a,b)]}.$
Lemma \ref{lem:variance} below gives us that
$\Var[Y(a,b)] \leq O(\eps^2)$, and the theorem is proved.
\end{proof}

It remains to establish the following lemma:

\begin{lemma} \label{lem:variance}
For each $1 \leq a,b \leq k$,
we have that $\Var[Y(a,b) ] = O(\eps^2).$
\end{lemma}
\begin{proof}
Fix $1 \leq a,b \leq k$, so $q_a(x_1,\dots,x_n)$
and $q_b(x_1,\dots,x_n)$ are degree-2 Gaussian polynomials
with mean 0.
Recalling the spherical symmetry of the $N(0,1)^n$ distribution,
by a suitable choice of basis that diagonalizes $q_a$ we may write
\[
q_a(x)=\littlesum_{i=1}^n \lambda_i x_i^2 +
\littlesum_{i=1}^n \beta_i x_i + \gamma
\quad \quad \text{and} \quad \quad
q_b(x)=\littlesum_{i,j=1}^n \delta_{ij}x_ix_j +
\littlesum_{i=1}^n \kappa_i x_i + \rho,
\]
where we take $\delta_{ij}=\delta_{ji}$ for all $1 \leq i,j \leq k.$

Recalling that $Y(a,b) = \langle D q_a, -DL^{-1}q_b \rangle$,
we start by observing that
$D q_a = (2 \lambda_\ell x_\ell + \beta_\ell)_{\ell=1,\dots,n}.$
For $-DL^{-1}q_b$, we have that
$
L^{-1} q_b  = - I_1(q_b)  - (1/2)I_2(q_b).
$
We have $I_1(q_b) = \littlesum_{i=1}^n \kappa_i x_i$.
Recalling that the first two normalized Hermite polynomials are
$h_1(x) = x$ and $h_2(x) = (x^2-1)/\sqrt{2}$,
it is straightforward
to verify that $I_2(q_b)$ (the homogeneous degree-2 part of the
Hermite expansion of $q_b$) is
\[
I_2(q_b) = \littlesum_{1 \leq i \neq j \leq k}
\delta_{ij} h_1(x_i) h_1(x_j) + \littlesum_{i=1}^n \sqrt{2} \cdot
\delta_{ii} h_2(x_i).
\]
Hence
\[
L^{-1} q_b  = - \littlesum_{i=1}^n \kappa_i x_i -
{\frac 1 2} \littlesum_{1 \leq i \neq j \leq k} \delta_{ij}
x_i x_j - {\frac 1 2} \littlesum_{i=1}^n \delta_{ii}(x_i^2 - 1),
\]
so
\[
-DL^{-1} q_b = \left(\kappa_\ell + \littlesum_{i=1}^n \delta_{i \ell}
 x_i\right)_{ \ell=1,\dots,n}.
\]
We thus can write $Y(a,b)$ as a degree-2 polynomial in the variables
$x_1,\dots,x_n$ as
\begin{eqnarray*}
Y(a,b) &=& \littlesum_{\ell=1}^n (2 \lambda_\ell x_\ell + \beta_\ell) \cdot
\left(\kappa_\ell + \littlesum_{i=1}^n \delta_{i \ell}x_i\right)\\
&=& \littlesum_{i=1}^n \littlesum_{\ell=1}^n 2 \lambda_\ell \delta_{i \ell}
x_i x_\ell + \littlesum_{\ell=1}^n 2 \kappa_\ell \lambda_\ell x_\ell
+ \littlesum_{i=1}^n \left(\littlesum_{\ell=1}^n \beta_\ell \delta_{i \ell}
\right)
x_i + \littlesum_{\ell=1}^n \kappa_\ell \beta_\ell.
\end{eqnarray*}
By Claim \ref{clm:variance-bound-SS}, we know that
$\Var[Y(a,b)] \leq SS(Y(a,b))$.  Using the inequality
$(r+s)^2 \leq 2r^2 + 2s^2$ for the degree-1 coefficients, to prove
the lemma it suffices to show that
\begin{equation} \label{eq:goal}
\littlesum_{i=1}^n \littlesum_{\ell=1}^n (\lambda_\ell \delta_{i \ell})^2
+ \littlesum_{\ell=1}^n (\kappa_\ell \lambda_\ell)^2 +
\littlesum_{i=1}^n \left(\littlesum_{\ell=1}^n \beta_\ell
\delta_{i \ell}\right)^2 \leq O(\eps^2).
\end{equation}
We bound each term of (\ref{eq:goal}) in turn.
For the first, we recall that each $\lambda_\ell$ is an eigenvalue
of $q_a$ and hence satisfies $\lambda_\ell^2 \leq \eps^2$; hence
we have
\[
\littlesum_{i=1}^n \littlesum_{\ell=1}^n (\lambda_\ell \delta_{i \ell})^2
\leq \eps^2 \littlesum_{i=1}^n \littlesum_{\ell=1}^n (\delta_{i\ell})^2
\leq \eps^2,
\]
where we have used Claim \ref{clm:variance-bound-SS} again to get
that $ \littlesum_{i,\ell=1}^n  (\delta_{i \ell})^2
\leq SS(q_b) \leq \Var[q_b] \leq 1.$
For the second term, we have
\[
\littlesum_{\ell=1}^n (\kappa_\ell \lambda_\ell)^2 \leq
\eps^2 \cdot \littlesum_{\ell=1}^n \kappa_\ell^2 \leq \eps^2 \cdot SS(q_b)
\leq \eps^2.
\]
Finally, for the third term, let us write $M=(\delta_{i\ell})$
for the $k \times k$ matrix corresponding to the quadratic part of
$q_b$ and $\bar{\beta}$ for the column
vector whose $\ell$-th entry is $\beta_\ell$.  Then we have
that
\[
\littlesum_{i=1}^n \left(\littlesum_{\ell=1}^n \beta_\ell
\delta_{i \ell}\right)^2 = \|M \bar{\beta}\|_2^2
\leq \|\lambda_{\max}(M) \bar{\beta}\|_2^2
\leq \eps^2 \|\bar{\beta}\|_2 \leq \eps^2,
\]
where the second inequality is because each eigenvalue of
$p_b$ has magnitude at most 1 and the third
is because $\|\bar{\beta}\|_2^2 \leq SS(p_a) \leq \Var[p_a] \leq 1.$
This concludes the proof of Lemma \ref{lem:variance}.
\end{proof}

\subsection{From test functions with bounded second derivative
to multidimensional Kolmogorov distance.} \label{sec:mollification}

In this subsection we show how ``mollification'' arguments can be used
to extend Theorem \ref{thm:h-match-cov} to
Theorem \ref{thm:small-dk}.
The main idea is to approximate the (discontinuous) indicator
function of an appropriate region by an appropriately ``mollified'' function
(that is continuous with bounded second derivative) so that the corresponding expectations
are approximately preserved. There are several different mollification constructions in the literature
that could potentially by used for this purpose. We use the following theore
from \cite{DKNfocs10}.
 \begin{theorem}\label{thm:smooth-apx}
[\cite{DKNfocs10}, Theorem 4.8 and Theorem 4.10]
 Let $I:\mathbb{R}^k \rightarrow \{0,1\}$ be the indicator of a region
$R$ in $\mathbb{R}^k$ and $c>0$ be arbitrary. Then there exists a function
$\tilde{I}_c:\mathbb{R}^k \rightarrow [0,1]$ satisfying:
 \begin{itemize}
 \item $\Vert \partial^{\beta} \tilde{I}_c / \partial
x^\beta\Vert_{\infty} \le (2c)^{|\beta|} $ for any $\beta \in \mathbb{N}^k$,
and
 \item $|I(x) - \tilde{I}_c(x)| \le \min \{1,
O((\frac{k}{c \cdot d(x,\partial R)})^2)\}$ for all $x \in \R^k$,
 \end{itemize}
 where $d(x,\partial R)$ is the Euclidean
distance of the point $x$ to the closest point in $R$.
 \end{theorem}

We use this to prove the following lemma, which
says that if a $k$-dimensional Gaussian $X$ ``mimics'' the
joint distribution $Y$ of a vector of $k$ degree-2 Gaussian polynomials
(in the sense of ``fooling'' all test functions $h$ with bounded second
derivative), then $X$ must have small $k$-dimensional Kolmogorov distance
from $Y$:

\begin{lemma} \label{lem:from-h-to-dk}
Let $p_1(x), \ldots, p_k(x) : \mathbb{R}^n \rightarrow \mathbb{R}$
be degree-2 polynomials with $\max_{i \in [k]} \Var(p_i) \ge \lambda$,
and let $Y$ be their joint distribution when $x$ is drawn from
$N(0,1)^n$.
Let $X \in \mathbb{R}^k$ be a jointly normal distribution such that
$\max_i \Var(X_i) \ge \lambda$. Suppose that for all functions
$h : \mathbb{R}^k \rightarrow \mathbb{R}, h \in {\cal C}^2$,
it holds that $|\mathbb{E} [h(X)] - \mathbb{E} [h(Y)]| \le
\Vert h'' \Vert_{\infty} \cdot \eta$. Then we have
$$
\dk(X,Y)
\le O\left( \frac{k^{1/3} \eta^{1/6}}{\lambda^{1/6}}\right).
 $$
\end{lemma}

\begin{proof}
Fix any $\theta \in \R^n$ and
define the function $I : \R^k \to \{0,1\}$
{to be the indicator of the region $R \eqdef
\{x \in \R^k: x_i \leq \theta_i\}.$}
Choose $c>0$.
We have
\begin{eqnarray*}
&&
\!\!\!\!\!\!\!\!\!\!
\!\!\!\!\!\!\!\!\!\!
\!\!\!\!\!\!\!\!\!\!
\big| \Pr [\forall \ i \in [k] \ X_i \le {\theta_i}]
- \Pr [\forall \ i \in [k] \ Y_i \le {\theta_i}]  \big| \\
&=&  \left| \mathbf{E} [I(X)] - \mathbf{E} [I(Y)]  \right| \\
&\le&
\left| \mathbf{E} [\tilde{I}_c(X)] - \mathbf{E} [\tilde{I}_c(Y)]  \right|
+ \left| \mathbf{E} [\tilde{I}_c(Y)] - \mathbf{E} I(Y)]  \right|
+ \left| \mathbf{E} [\tilde{I}_c(X)] - \mathbf{E} [I(X)]  \right| \\
&\le&
4c^2 \eta + \left| \mathbf{E} [\tilde{I}_c(Y)] - \mathbf{E} I(Y)] \right|
+ \left| \mathbf{E} [\tilde{I}_c(X)] - \mathbf{E} [I(X)]  \right|,
\end{eqnarray*}
where we used the first item of Theorem~\ref{thm:smooth-apx} to bound
the first term.
We proceed to bound the other two terms.
For the first one, choose $\delta>0$ and now note that
\begin{eqnarray*}
\big| \mathbf{E} [\tilde{I}_c(Y)] - \mathbf{E} I(Y)]  \big|
&\le&  \mathbf{E}_{y \sim Y} [|\tilde{I}_c(y) - I(y)|]  \\
&\le& \Pr_{y \sim Y} [d(y, \partial R) \le \delta] + O
\left( \frac{k^2}{c^2 \delta^2} \right) \\
&\le& O \left(\frac{\sqrt{\delta}}{\lambda^{1/4}}\right) +
O \left( \frac{k^2}{c^2 \delta^2} \right),
\end{eqnarray*}
The second inequality above
used $0 \leq I, \tilde{I}_c \leq 1$ and the second item of Theorem
\ref{thm:smooth-apx}.  The
final inequality used the Carbery-Wright anti-concentration
bound (Theorem \ref{thm:cw})
together with the observation that in order for $y \sim Y$ to be
within distance $\delta$ of $\delta R$, it must be the case
that $|p_i(y) - \theta_i| \leq \delta$ where $i$ is the element of $[k]$ that
has $\Var(p_i) \geq \lambda$.
Similar reasoning gives that
$$
\big| \mathbf{E} [\tilde{I}_c(X)] - \mathbf{E} I(X)]  \big| \le
O \left(\frac{\sqrt{\delta}}{\lambda^{1/4}}\right)
+ O \big( \frac{k^2}{c^2 \delta^2} \big)
$$
(in fact here the $\frac{\sqrt{\delta}}{\lambda^{1/4}}$ can be strengthened
to ${\frac {\delta}{\lambda^{1/2}}}$ because now $X_i$ is a degree-1
rather than degree-$2$ polynomial in $N(0,1)$ Gaussians, but this
will not help the overall bound).
Optimizing for $\delta$ by setting $\delta = k^{4/5} \lambda^{1/10}
/c^{4/5}$, we get that
 $$
\big| \Pr [\forall \ i \in [k] \ X_i \le \theta_i]
- \Pr [\forall \ i \in [k] \ Y_i \le \theta_i]
\big| \le 4c^2 \eta + O\left(\frac{k^{2/5}}{c^{2/5} \lambda^{1/5}} \right).
 $$

Now optimizing for $c$ by choosing $c = k^{1/6}/(\eta^{5/12}\gamma^{1/12})$, we get that
 $$
 \big| \Pr [\forall \ i \in [k] \ X_i \le \theta_i]
- \Pr [\forall \ i \in [k] \ Y_i \le \theta_i]
\big| \le O\left(\frac{k^{1/3} \eta^{1/6}}{\lambda^{1/6}}\right),
 $$
which concludes the proof of Lemma \ref{lem:from-h-to-dk}.
 \end{proof}

With Lemma \ref{lem:from-h-to-dk} and Theorem \ref{thm:h-match-cov}
in hand we are ready to prove Theorem \ref{thm:small-dk}:

\medskip

\noindent \emph{Proof of Theorem \ref{thm:small-dk}:}
For $i \in [k]$ let $\tilde{q}_i(x) = q_i(x) - \E[q_i]$,
so $\tilde{q}_i$ has mean zero.
Applying Theorem \ref{thm:h-match-cov} to
$\tilde{q}=(\tilde{q}_1,\dots,\tilde{q}_k)$ we get
that any $h$ with $\|h''\|_{\infty} \leq \infty$ satisfies
$|\E[h(\tilde{q})] - \E[h(N(0,C))]| \leq O(k^2 \eps) \cdot
\|h''\|_{\infty}.
$
Applying Lemma \ref{lem:from-h-to-dk}, taking $X$ to be $N(0,C)$
and its $\eta$ parameter to be $O(k^2 \eps)$,
we get that
\[
\dk(\tilde{q},N(0,C)) \
\leq O\left({\frac {k^{2/3} \eps^{1/6}}{\lambda^{1/6}}}\right),
\]
which gives the theorem as claimed.
\qed

\fi

\section{Transforming a $k$-tuple of degree-2 Gaussian polynomials}
\label{sec:gauss-decomp}

In this section we present a deterministic procedure,
called {\bf Transform}, which transforms
an arbitrary $k$-tuple of degree-2 polynomials $(q_1,\dots,q_k)$
into an ``essentially equivalent'' (for the purpose of approximately
counting PTF satisfying assignments under the Gaussian distribution)
$k$-tuple of degree-2 polynomials $(r_1,\dots,r_k)$
that have a ``nice structure''. This structure enables an efficient
deterministic decomposition of the joint distribution.  In the following section
we will give an efficient algorithm to do deterministic approximate
counting for vectors of polynomials with this ``nice structure.''

In more detail, the main theorem of this section, Theorem \ref{thm:transform},
says the following:
Any $k$-tuple $q=(q_1,\dots,q_k)$ of degree-2
Gaussian polynomials can be efficiently deterministically transformed into
a $k$-tuple $r=(r_1,\dots,r_k)$ of degree-2 Gaussian polynomials such
that (i) $\dk(r,q) \leq O(\eps)$, and (ii) \ignore{
with probability at least $1-\eps$ over a
random restriction fixing the first $t=\poly(k/\eps)$ variables}for
every restriction fixing the first $t=\poly(k/\eps)$ variables,
\ignore{
(i.e. over a random $\rho=(\rho_1,\dots,\rho_t) \sim N(0,1)^t)$),}
the $k$-tuple $r|_\rho=(r_1|_\rho,\dots,r_k|_\rho)$
of restricted polynomials has
$k$-dimensional Kolmogorov distance $O(\eps)$
from the $k$-dimensional Normal distribution with matching mean
and covariance matrix.  More formally,

\begin{theorem} \label{thm:transform}
There is an algorithm {\bf Transform} with the following
properties:  It takes as input a
$k$-tuple $q=(q_1,\dots,q_k)$ of degree-2 polynomials over $\R^n$
with $\Var_{x \sim N(0,1)^n}[q_i(x)]=1$ for all $i \in [k]$,
and a parameter $\eps>0$.  It runs in deterministic
time $\poly(n, k, 1/\eps)$ and outputs a $k$-tuple
$r=(r_1,\dots,r_k)$ of degree-2 polynomials over $\R^n$
and a value $0 \leq t \leq O(k \ln(1/\eps)/\eps^2)$
such that both of the following hold:

\begin{enumerate}

\item [(i)] $\dk(q,r) \leq O(\eps)$,
where $q$ is the random variable
$q=(q_1(x),\dots,q_k(x))$ with $x \sim N(0,1)^n$ and
$r=(r_1(y),\dots,r_k(y))$ with $y \sim N(0,1)^n$;
and

\item [(ii)] \ignore{
With probability $1-\eps$ over a random
$\rho=(\rho_1,\dots,\rho_t) \sim N(0,1)^{t}$,}For every restriction
$\rho=(\rho_1,\dots,\rho_t)$, we have
 \ifnum\confversion=1
$\dk(r|_\rho,N(\mu(r|_\rho),\Sigma(r_\rho))) \leq \eps.$
 \else
\[
\dk(r|_\rho,N(\mu(r|_\rho),\Sigma(r_\rho)))
\leq \eps.
\]
\fi
Here ``$r_\rho$'' denotes the random variable
$(r_1|_\rho(y),\dots,r_k|_\rho(y))$ where $y \sim N(0,1)^n$
and $r_i|_\rho(y) \eqdef r_i(\rho_1,\dots,\rho_t,y_{t+1},\dots,y_n)$;
$\mu(r|_\rho)$ denotes the vector of means $(\mu_1|_\rho,\dots,
\mu_k|_\rho) \in \R^k$
where $\mu_i|_\rho = \E_{y \sim N(0,1)^n}[r_i|_\rho(y)]$;
and $\Sigma(r_\rho)$ denotes the covariance matrix in $\R^{k \times k}$
whose $(i,j)$ entry is \ifnum\confversion=0 \\ \fi $\Cov_{y \sim N(0,1)^n}(r_i|_\rho(y),r_j|_\rho(y)).$

\end{enumerate}

\end{theorem}

At a high level,
the {\bf Transform} procedure first performs a ``change of basis''
using the procedure {\bf Change-Basis}
to convert $q=(q_1(x),\dots,q_k(x))$ into an ``almost equivalent''
vector $p=(p_1(y),\dots,p_k(y))$ of polynomials.  (Conceptually
the distribution of $(p_1(y),\dots,p_k(y))$ is identical
to the distribution of $(q_1(x),\dots,q_k(x))$, but in reality
some approximations need to be made because we can only
approximately compute eigenvalues, etc.; hence the two vector-valued
random variables are only ``almost equivalent.'')
Next, the {\bf Transform} procedure runs {\bf Process-Polys}
on $(p_1,\dots,p_k)$; this further changes each $p_i$ slightly,
and yields polynomials $r_1,\dots,r_k$ which are the final output
of {\bf Transform}$(q_1,\dots,q_k).$
A detailed description of the {\bf Transform} procedure follows:

\begin{framed}
\noindent {\bf Transform}

\smallskip

\noindent {\bf Input:}  vector $q=(q_1,\dots,q_k)$ of degree-2
polynomials $q_\ell(x_1,\dots,x_n)$ such that $\E_{x \sim N(0,1)^n}
[q_{\ell}(x)^2]=1$ for all $\ell=1,\dots,k$; parameter $\eps>0$

\noindent {\bf Output:} A vector $r=(r_1(y),\dots,r_k(y))$
of degree-2 polynomials over $\R^n$,
and a value $0 \leq t \leq O(k \ln(1/\eps)/\eps^2)$.
\ignore{
and a value $0 \leq k' \leq k$, satisfying the guarantees
of Theorem \ref{thm:transform}.
}

\medskip

\begin{enumerate}

\item Set $\eta = (\eps/k)^4/(\log(k/\eps))^2$
and $\epsa=\eps^{12}\eta^2/k^8.$

\item Run {\bf Change-Basis}$((q_1,\dots,q_k),\epsa,\eta)$
and let $(p_1,\dots,p_k),t$ be its output.


\item Run {\bf Process-Polys}$((p_1,\dots,p_k),t,\eta)$
and let $(r_1,\dots,r_k),k'$ be its output.

\item Output $(r_1,\dots,r_k),t$.

\end{enumerate}

\end{framed}

\ifnum\confversion=0
Subsection \ref{sec:changebasis} below gives a detailed description
and analysis of {\bf Change-Basis}, Subsection \ref{sec:processpolys}
does the same for {\bf Process-Polys},
and Subsection \ref{sec:transform-proof}
proves Theorem \ref{thm:transform}.
\else
Subsection \ref{sec:changebasis} below gives a description
and analysis sketch of {\bf Change-Basis}. Further details are postponed to the full version.
\fi

\subsection{The {\bf Change-Basis} procedure.}
\label{sec:changebasis}

\ifnum\confversion=0
\paragraph{Intuition.}
The high-level approach of the {\bf Change-Basis} procedure is similar to the
decomposition procedure for vectors of $k$ linear
forms that was given in \cite{GOWZ10}, but there are significant additional
complications that arise in our setting.  
Briefly, in the \cite{GOWZ10} approach, a vector of $k$ linear forms 
is simplified
by ``collecting'' variables in a greedy fashion.  Each of the
$k$ linear forms has a budget of at most $B$, meaning that at most
$B$ variables will be collected on its behalf; thus the overall number
of variables that are collected is at most $kB$.  Intuitively,
at each stage some variable is collected which has large influence in the
remaining (uncollected) portion of some linear form.  The \cite{GOWZ10}
analysis shows that
after at most $B$ variables have been collected on behalf of each linear form, 
each of the $k$ linear forms will either be regular or its remaining portion
(consisting of the uncollected variables) will have small variance.
(See Section \ref{sec:prev-res}
for a more detailed overview of the \cite{GOWZ10} decomposition procedure).

In our current setting, we are dealing with $k$ degree-$2$ Gaussian polynomials
instead of $k$ linear forms, and linear forms will play a role for us
which is analogous to the role that single variables played in \cite{GOWZ10}.  
Thus each quadratic polynomial will have at most $B$ linear forms
collected on its behalf and at most $kB$ linear forms will be collected
overall.
Of course \emph{a priori} there are uncountably many possible 
linear forms to contend with, so it is not clear how to select a single
linear form to collect in each stage.  We do this by (approximately)
computing the largest eigenvalues of each quadratic form; in each stage
we collect some linear form, corresponding to an eigenvector for
some quadratic polynomial, whose corresponding eigenvalue is large
compared to the variance of the remaining (``uncollected'')
portion of the quadratic polynomial.
An argument similar to that of \cite{GOWZ10} shows that after at most
$B$ linear forms have been collected on behalf of each quadratic
polynomial, each of the $k$ quadratic polynomials will either be
``regular'' (have small largest eigenvalue compared to the variance of the
remaining portion), or else the variance of the remaining portion will
be small.

\begin{remark}
In this section we describe an ``idealized'' version of the algorithm which assumes
that we can do certain operations (construct an orthonormal basis,
compute eigenvalues and eigenvectors) exactly with no error.
In fact these operations can only be carried out approximately, but the errors
can in all cases be made extremely small so that running the algorithm
with ``low-error'' implementations of the idealized steps still gives
a successful implementation overall.  However, keeping track of the errors
and approximations is quite cumbersome, so in order to highlight the main
ideas we begin by describing the ``idealized'' version.

We will try to clearly state all of the idealized assumptions
as they come up in the idealized algorithm below.
We will state Lemma \ref{lem:change-basis}, the main lemma about the
{\bf Change-Basis} algorithm, in versions
corresponding both to the ``idealized''
algorithm and to the ``real'' algorithm.
\end{remark}
\fi

\subsubsection{Setup for the {\bf Change-Basis} procedure.}

We start with a few definitions.  We say that a set ${\cal A}
= \{L_1(x),\dots,L_r(x)\}$ of $r \leq n$ linear forms
$L_i(x) = v^{(i)} \cdot x$ over $x_1,\dots,x_n$
is \emph{orthonormal} if $\E_{x \sim N(0,1)^n}[L_i(x)L_j(x)] =
\delta_{ij}$ for $1 \leq i,j \leq r$ (equivalently, $v^{(1)},\dots,
v^{(r)}$ are orthonormal vectors).

%
%
\begin{definition}
Let $q: \R^n \rightarrow \R$ be a degree-$2$ polynomial
\ifnum\confversion=1
$q(x) = \littlesum_{1\leq i \leq j \leq n} a_{ij} x_ix_j + \littlesum_{1\leq i \leq n}
b_i x_i + c$ and let
\else 
\begin{equation}
\label{eq:qorig}
q(x) = \sum_{1\leq i \leq j \leq n} a_{ij} x_ix_j + \sum_{1\leq i \leq n}
b_i x_i + c,
\end{equation} 
and let
\fi
$\{L_i(x)=v^{(i)} \cdot x\}_{i=1,\dots,n}$ be a full orthonormal set
of linear forms.  Let ${\cal A} = \{L_1,\dots,L_r\}$
and ${\cal B} = \{L_{r+1},\dots,L_n\}$ for some $0 \leq r \leq n.$
We define
$\proj(q,{\cal A},{\cal B})$,
the \emph{projection of $q$ onto ${\cal A}$}, and
$\res(q,{\cal A},{\cal B})$, the
\emph{residue of $q$ w.r.t. ${\cal A}$}, as follows.
Rewrite $q$ using the linear forms $L_i(x)$, i.e.
\begin{equation}
q = \littlesum_{1 \leq i \leq j \leq n} \alpha_{ij} L_i(x) L_j(x) +
\littlesum_{1 \leq i \leq n} \beta_i L_i(x) + c.
\label{eq:qnew}
\end{equation}
Define
\begin{equation} \label{eq:res}
\res(q, {\cal A},{\cal B})
\eqdef
 \littlesum_{r < i \le j \le n} \alpha_{ij} L_i(x) L_j(x)
+ \littlesum_{r < i \le n} \beta_i L_i(x) + c
\end{equation}
and
\[
\proj(q,{\cal A},{\cal B}) \eqdef q -
\res(q, {\cal A},{\cal B}).
\]
\end{definition}

\ignore{
\begin{fact} \label{fact:well-defined}
$\proj(p,\{L_1, \ldots, L_r\})$
and
$\res(p,\{L_1,\dots,L_r\})$ are uniquely defined regardless of the
specific choice of linear forms $L_{r+1}(x),\dots,L_{n}(n)$ that complete
${\cal A}$ to a full orthonormal basis.
\end{fact}
\begin{proof}
First note that given a specific orthonormal set $L_{1}(x),\dots,L_n(x)$
of linear forms, each $x_i$ can be uniquely expressed as a linear
combination $\sum_{j=1}^n u_{ij} L_j(x)$ of $L_1(x),\dots,L_n(x)$
and hence $p$ may be uniquely rewritten as in (\ref{eq:qnew}) by
substituting $\sum_{j=1}^n u_{ij} L_j(x)$ for each occurrence of
$x_i$ in (\ref{eq:qorig}).

\new{{\bf To do:}  prove the fact.}
\end{proof}
}

\ifnum\confversion=0
\noindent Note that the residue (resp. projection) of $q$ corresponds to the tail (resp. head) of $q$ in the basis of the linear forms $L_i$.

\medskip
\fi

\noindent
{\bf Idealized Assumption \#1:}  There is a $\poly(n)$ time deterministic
procedure {\bf Complete-Basis}
which, given a set ${\cal A}=\{L_i(x)\}_{i=1,\dots,r}$ of
orthonormal linear forms, outputs a set ${\cal B}=\{L_{j}(x)\}_{j=r+1,\dots,n}$
such that ${\cal A} \cup {\cal B}$ is a full orthonormal set of linear forms.

\medskip

\begin{claim} \label{claim:new-coeff}
There is an efficient algorithm {\bf Rewrite}
which, given as input $q$ \ifnum\confversion=0(in the form (\ref{eq:qorig})\fi
and sets ${\cal A}=\{L_i(x)\}_{i=1,\dots,r}$,
${\cal B}=\{L_{r+1}(x),\dots,L_n(x)\}$ such that
${\cal A} \cup {\cal B}$ is a full orthonormal basis, outputs
coefficients $\alpha_{ij}, \beta_i, c$ such that (\ref{eq:qnew})
holds.
\end{claim}
\ifnum\confversion=0
\noindent {\bf Proof sketch:}
Given ${\cal A}$ and ${\cal B}$
by performing a matrix inversion it is possible to efficiently
compute coefficients $u_{ij}$ such that for $i \in [n]$
we have $x_i = \sum_{j=1}^n u_{ij} L_j(x).$  Substituting
$\sum_{j=1}^n u_{ij} L_j(x)$ for each occurrence of $x_i$ in
(\ref{eq:qnew}) we may rewrite $q$ in the form (\ref{eq:qnew})
and obtain the desired coefficients.
\qed
\fi

\ignore{
Note that given these $\alpha_{ij},\beta_i,c$ coefficients it is
straightforward to write down the polynomials
$\res(q, {\cal A},{\cal B})$ (see (\ref{eq:res})
and
$\proj(q,{\cal A},{\cal B})  q -
\res(q, {\cal A},{\cal B})$.
}

\medskip


Next we observe that the largest eigenvalue can never increase as we
consider the residue of $q$ with respect to larger and larger orthonormal sets
of linear forms:
\begin{lemma}
Fix any degree-2 polynomial $q$ and any full orthonormal set
$\{L_i(x)=v^{(i)} \cdot x\}_{i=1,\dots,n}$
of linear forms.
Let
${\cal A}= \{L_i(x)=v^{(i)} \cdot x\}_{i=1,\dots,r}$
and
${\cal B}= \{L_i(x)=v^{(i)} \cdot x\}_{i=r+1,\dots,n}$.
Then we have that $|\lambda_{\max} (\res(q,{\cal A},{\cal B})) |
\leq |\lambda_{\max} (q)|$.
\end{lemma}
\ifnum\confversion=0
\begin{proof} 
Let $M$ be the $n \times n$ symmetric matrix corresponding to the quadratic
part of $q$, and let $M'$ be the $n \times n$ symmetric matrix
corresponding to the quadratic part of $\res(q,{\cal A},{\cal B}).$
Let $\tilde{M}$ be the symmetric matrix obtained from $M$ by a change of
basis to the new coordinate system defined by the $n$ orthonormal
linear forms $L_1,\dots,L_n$, and likewise let $\tilde{M}'$
be the matrix obtained from $M'$ by the same change of basis.
Note that $\tilde{M}'$ is obtained from $\tilde{M}$ by zeroing out all
entries $\tilde{M}_{ij}$ that have either $i \in {\cal A}$
or $j \in {\cal A}$, i.e. $\tilde{M}'$ corresponds to the principal minor
$\tilde{M}_{{\cal B},{\cal B}}$ of $\tilde{M}.$
Since eigenvalues are unaffected by a change of basis, it suffices to
show that
$|\lambda_{\max}(\tilde{M})| \geq |\lambda_{\max}(\tilde{M}')|.$

We may suppose without loss of generality that 
$\lambda_{\max}(\tilde{M}')$ is positive.  
By the variational characterization of eigenvalues we have that
$\lambda_{\max}(\tilde{M}') = \max_{\|x\|=1} x^{T} \tilde{M}' x$.  
Since $\tilde{M}'$ corresponds to the principal minor
$\tilde{M}_{{\cal B},{\cal B}}$ of $\tilde{M}$, a vector $x'$ 
that achieves the maximum must have nonzero coordinates only
in ${\cal B}$, and thus 
\[
\lambda_{\max}(\tilde{M}') = (x')^T \tilde{M}' x' = 
(x')^T \tilde{M} x' \leq \max_{\|x\|=1} x^{T} \tilde{M}' x
\leq 
|\lambda_{\max}(\tilde{M})|.
\]
\end{proof}
\fi
\ignore{

\begin{fact}\label{fac:11}
For any symmetric matrix $A$ and $I \subseteq [n]$, 
let $A_I$ be the minor of $A$ corresponding to the set $I$. 
If $\sigma_1(A)$ denotes the largest singular value of $A
$, the note that $\sigma_{1} (A_I) \le \sigma_1(A)$.
\end{fact}
\begin{proof}
Let $B = A^T \cdot A$ and $B ' =A_I^T \cdot A_I$.  It is easy to see that
$B_I = B'$.\rnote{This seems wrong --
unless I have lost my mind, the 2-by-2 symmetric matrix with entries
1,2,2,3 and $I=\{1\}$ is a counterexample to this.}
The proof now follows from the definition of Rayleigh quotient.
\end{proof}
\begin{corollary}
For any degree-2 polynomial $p$ and an orthonormal set $L_1, \ldots, L_k$, $\lambda_{\max} (\mathop{Res}_{L_1, \ldots, L_k}(p)) \le \lambda_{\max} (p)$.
\end{corollary}
\begin{proof}
Obvious from Fact~\ref{fac:11}.
\end{proof}
}

\subsubsection{The {\bf Change-Basis} procedure.}

We now describe the {\bf Change-Basis} procedure.
This procedure takes as input a vector $q=(q_1,\dots,q_k)$ of $k$
degree-2 polynomials, where each $q_i$ is specified explicitly by its
coefficients\ifnum\confversion=0 as in (\ref{eq:qorig})\fi,
and two parameters $\epsa,\eta>0$.
It outputs a vector of polynomials $p=(p_1(y),\dots,p_k(y))$
where each $p_\ell(y_1,\dots,y_n)$ is also specified
explicitly by coefficients
$\alpha^{(\ell)}_{ij}$, $\beta^{(\ell)}_i$, $c^{(\ell)}$
that define $p_\ell(y)$
as
\begin{equation} \label{eq:pell}
p_\ell(y) = \littlesum_{1 \leq i \leq j \leq n} \alpha^{(\ell)}_{ij} y_i y_j
+ \littlesum_{1 \leq i \leq n} \beta^{(\ell)}_i y_i + c^{(\ell)},
\end{equation}
and an integer $0 \leq t \leq k \ln(1/\eta)/\epsa^2$.
As its name suggests, the {\bf Change-Basis} procedure
essentially performs a change of basis on $\R^n$ and
rewrites the polynomials $q_{\ell}(x)$
in the new basis as $p_{\ell}(y)$.
It is helpful to think of $y_i$ as playing the role of
$L_i(x)$ where $\{L_i(x)\}_{i=1,\dots,n}$ is
a set of orthonormal linear forms computed by the algorithm,
and to think of the coefficients
$\alpha^{(\ell)}_{ij}$, $\beta^{(\ell)}_i$, $c^{(\ell)}$
defining $p_\ell(y)$ as being obtained from $q_\ell(x)$ by rewriting
$q_\ell(x)$ using the linear forms $L_i(x)$ as in (\ref{eq:qnew}).

The {\bf Change-Basis} procedure has two key
properties.  The first is that
the two vector-valued random variables $(q_1(x),\dots,q_k(x))$
(where $x \sim N(0,1)^n$) and $(p_1(y),\dots,p_k(y))$
(where $y \sim N(0,1)^n$) are very close in Kolmogorov distance.
(In the ``idealized'' version they are identically distributed,
and in the ``real'' version they are close in $k$-dimensional
Kolmogorov distance.)
The second is that each of the $p_\ell$ polynomials is ``nice''
in a sense which we make precise in Lemma \ref{lem:change-basis}
below.  (Roughly speaking, $p_\ell$ either almost entirely depends only
on a few variables, or else has a small-magnitude max eigenvalue.)

\begin{framed}
\noindent {\bf Change-Basis}

\smallskip

\noindent {\bf Input:}  vector $q=(q_1,\dots,q_k)$ of degree-2
polynomials $q_\ell(x_1,\dots,x_n)$ such that $\E_{x \sim N(0,1)^n}
[q_\ell(x)^2]=1$ for all $\ell=1,\dots,k$; parameters $\epsa,\eta>0$


\noindent {\bf Output:} A vector $p=(p_1(y),\dots,p_k(y))$
of degree-2 polynomials (described explicitly via their coefficients
as in (\ref{eq:pell})) satisfying the guarantees of
Lemma \ref{lem:change-basis}, and an integer $t \geq 0.$

\medskip

\begin{enumerate}

\item Initialize the set of linear forms ${\cal A}$ to be
$\emptyset$.  Let $\tilde{q}_\ell(x)=q_\ell(x)$
for all $\ell=1,\dots,k.$

\item If each $\ell=1,\dots,k$ is such that $\tilde{q}_\ell$ satisfies
either

\begin{itemize}
\item [(a)] $\Var[\tilde{q}_\ell] \leq \eta$, ~~~~~~~~~~~~ or
~~~~~~~~~~~~
(b) ~~ ${\frac {(\lambda_{\max}(\tilde{q}_\ell))^2}{\Var[\tilde{q}_\ell]]}}
\leq \epsa$,
\end{itemize}
then use {\bf Complete-Basis} to
compute a set ${\cal B}$ of linear forms ${\cal B}=
\{L_{|{\cal A}|+1}(x),\dots,L_n(x)\}$ such that
${\cal A} \cup {\cal B}$ is a full orthonormal basis,
and go to Step 5.  Otherwise, proceed to Step 3.

\item Let $\ell' \in [k]$ be such that $\tilde{q}_{\ell'}$ does not satisfy
either (a) or (b) above.
Let $v \in \R^n$ be a unit eigenvector corresponding to
the maximum magnitude eigenvalue $\lambda_{\max}(\tilde{q}_{\ell'})$.
Let $L(x) = v \cdot x$.  Add $L(x)$ to ${\cal A}$.

\item
Use {\bf Complete-Basis}$({\cal A})$ to
compute a set of linear forms
${\cal B} = \{L_{|{\cal A}|+1}(x),\dots,L_n(x)\}$ such that
${\cal A} \cup {\cal B}$ is a full orthonormal basis.
For all $\ell = 1,\dots,k$ use {\bf Rewrite}$(q_\ell,{\cal A},
{\cal B})$ to compute coefficients $\alpha^{(\ell)}_{ij},\beta^{(\ell)}_i$,
$c^{(\ell)}$ as in (\ref{eq:qnew})).  Set
$\tilde{q}_\ell(x) = \res(q_\ell,{\cal A}, {\cal B})$
and $\proj(q_\ell,{\cal A}, {\cal B})=q_\ell(x) - \tilde{q}_\ell(x).$
Go to Step~2.


\item We have ${\cal A} = \{L_1(x),\dots,L_{|{\cal A}|}(x)\}$
and ${\cal B}=\{L_{|{\cal A}|+1}(x),\dots,L_n(x)\}$.  For each $\ell \in [k]$
use {\bf Rewrite} on $q_\ell$ to compute coefficients
$\alpha^{(\ell)}_{ij}$, $\beta^{(\ell)}_i$, $c^{(\ell)}$
such that
\[
q_\ell(x) =
\littlesum_{1 \leq i \leq j \leq n} \alpha^{(\ell)}_{ij}L_i(x)L_j(x) +
\littlesum_{1 \leq i \leq n} \beta^{(\ell)}_{i}L_i(x) +
c^{(\ell)}.
\]
Output the polynomials $p_1(y),\dots,p_k(y)$ defined by these coefficients
as in (\ref{eq:pell}), and the value $t=|{\cal A}|.$

\end{enumerate}
\vspace{-0.2cm}
\end{framed}

{\bf Idealized assumption \#2:}  There is a $\poly(n)$ time deterministic procedure
which, given $\tilde{q}_\ell$ as input,
\begin{itemize}
\item exactly computes the maximum eigenvalue $\lambda_{\max}(\tilde{q}_\ell)$,
and
\item exactly computes a unit eigenvector corresponding to
$\lambda_{\max}(\tilde{q}_\ell)$.
\end{itemize}

\ifnum\confversion=0
\ignore{
\green{
It is instructive to consider the above definition from a matrix
perspective.   Let $A=(A_{ij})$ be the symmetric matrix corresponding to the
quadratic part of $q$ as described in Section \ref{sec:prelim}, i.e.
$A_{ij} = a_{ij} (1/2 + \delta_{ij}/2)$.  By the spectral theorem for
real symmetric matrices, we can write $A$ as
$\sum_{j=1}^n \lambda_j (v^j)^T \cdot v^j$ where $(\lambda_1,v^1),\dots,
(\lambda_n,v^n)$ are the $n$ eigenvalue-eigenvector pairs in the
spectral decomposition.
}
}

Before we proceed with the proof, we recall some basic facts: 

\begin{definition}[Rotational invariance of polynomials]
Given two polynomials $p(x) = \littlesum_{1 \le i \le j \le n} a_{ij}x_i x_j
+ \littlesum_{1\le i \le n} b_i x_i +C$ and $q(x) =\littlesum_{1 \le i
\le j \le n} a'_{ij}x_i x_j + \littlesum_{1\le i \le n} b'_i x_i +C$
with the same constant term, we say that they are
\emph{rotationally equivalent} if there is an orthogonal matrix $Q$ such
that $Q^T \cdot A \cdot Q = A'$ and $ Q^T \cdot b = b'$.
If the matrix $A'$ is diagonal then the polynomial $q$ is said to be the
\emph{decoupled equivalent of $p$}.
In this case, the eigenvalues of $A$ (or
equivalently $A'$) are said to be the eigenvalues of the quadratic form $p$.
\end{definition}

\begin{claim}\label{clm:poly-equivalence}
For any degree-$2$ polynomials $p(x)$ and $q(x)$ which are rotationally equivalent, the distributions of $p(x)$ and $q(x)$ are identical when $(x_1, \ldots, x_n) \sim N(0,1)^n$.
\end{claim}

\noindent For $x \sim N(0,1)^n)$, since $L_1, \ldots, L_n$ is an orthonormal basis, we have that $(L_1(x),\dots,L_n(x))$
is distributed identically to $(y_1,\dots,y_n) \sim N(0,1)^n$.
By construction, we have that the matrix corresponding to $p_\ell$ is an
orthogonal transformation of the matrix corresponding to $q_\ell$. That is, $p_{\ell}$ and $q_{\ell}$ are rotationally equivalent.

Recalling the $\tail_t(\cdot)$ and $\head_t(\cdot)$
notation from Section \ref{sec:prelim},
we see that the polynomial $\tail_t(p_\ell(y))$ corresponds
precisely to the polynomial $\res(q_\ell,{\cal A},{\cal B})$
and that $\head_t(p_\ell(y))$ corresponds precisely to $\proj(q_\ell,
{\cal A},{\cal B}).$  As a consequence, the eigenvalues
of $\tail_t(p_\ell)$ are identical to the eigenvalues of
$\tilde{q}_\ell$.

\begin{claim}
Let $q(x)$ be a degree-2 Gaussian
polynomial and ${\cal A}=\{L_1(x),\dots,L_r(x)\}$ be
an orthonormal set of linear forms.  Let $\tilde{q}_\ell(x)=
\res(q_\ell,{\cal A},{\cal B})$
and let $v$ be a unit eigenvector of (the symmetric matrix corresponding to)
$\tilde{q}_\ell.$  Then the
linear form $L'(x) \eqdef v \cdot x$ is orthogonal
to all of $L_1,\dots,L_r$, i.e., $\E[L'(x) \cdot L_i(x)]=0$ for all
$\ell=1,\dots,r.$
\end{claim}
\begin{proof}
The claim follows from the fact that $v$ lies in the span of ${\cal B}$ (as follows by the definition of the residue) and that the sets of vectors ${\cal A}$ and ${\cal B}$ are orthonormal.
\end{proof}

The above claim immediately implies that throughout the execution of
{\bf Change-Basis}, ${\cal A}$ is always an orthonormal
set of linear forms:
\begin{corollary}
At every stage in the execution of {\bf Change-Basis},
the set ${\cal A}$ is orthonormal.
\end{corollary}

(As a side note we observe that since ${\cal A} \cup {\cal B}$
is a full orthonormal set, it is indeed straightforward to compute
$\Var[\tilde{q}_\ell]$ in Step 2;
the first time Step 2 is reached this is simply the
same as $\Var[q_\ell]$, and in subsequent iterations
we can do this in a straightforward way
since we have computed the coefficients
$\alpha^{(\ell)}_{ij},\beta^{(\ell)}_i$ in Step 4 immediately before
reaching Step 2.)

Next we bound the value of $t$ that the algorithm outputs:

\begin{claim}
The number of times that {\bf Change-Basis} visits Step~2 is at most
$k \ln (1/\eta)/\epsa^2 $.  Hence the value $t$ that the algorithm
returns is at most
$k \ln (1/\eta)/\epsa^2 $.
\end{claim}
\begin{proof}
It is easy to see that after the end of any iteration,
for any fixed $\ell \in [k]$, the variance of $\tilde{q}_{\ell}$ does not
increase.  This follows by the definition of the residue and the expression of the variance as a function of the coefficients.
At the start of the algorithm each $\tilde{q}_{\ell}$ has
$\Var(\tilde{q}_{\ell}) = 1$.  We claim that each time Step 3 is reached, the polynomial $\tilde{q}_{i'}$, $i' \in [k]$,
that is identified in that step has its variance
$\Var(\tilde{q}_{i'})$ multiplied by a value which is at most
$(1- \Omega(\epsa^2))$ in the corresponding iteration. The claim follows from the fact that
the maximum magnitude eigenvalue of $\tilde{q}_{i'}$ is at least $\epsa \cdot \sqrt{\Var[\tilde{q}_{i'}]}$ and the definition of the residue.
Thus each specific $j \in [k]$ can be chosen as the $i'$
in Step 3 at most $O(\ln(1/\eta)/\epsa^2)$ times (after this many iterations
it will be the case that $\Var[\tilde{q}_j] \leq \eta$).
This proves the claim.
\end{proof}

Thus we have proved the following:
\fi
\ifnum\confversion=1 In the full version, we prove the following:\fi
\begin{lemma} \label{lem:change-basis}
(Idealized lemma about {\bf Change-Basis}:)
Given as input a vector $q=(q_1,\dots,q_k)$ of degree-2 polynomials
such that $\E_{x \sim N(0,1)^n}[q_i(x)^2]=1$
and parameters $\epsa,\eta>0$,
the algorithm {\bf Change-Basis}$((q_1,\dots,q_k),\epsa,\eta))$
runs in time $\poly(n, t, 1/\eps')$ and outputs polynomials $p_1(y),\dots,p_k(y)$
(described via their coefficients as in (\ref{eq:pell}))
and a value $0 \leq t \leq k \ln(1/\eta)/\epsa^2$
such that items (1) and (2) below both hold.

\begin{enumerate}

\item The
vector-valued random variables $q=(q_1(x),\dots,q_k(x))$
(where $x \sim N(0,1)^n$) and $p=(p_1(y),\dots,p_k(y))$
(where $y \sim N(0,1)^n$) are identically distributed.

\item For each $\ell \in [k]$, at least one of the following holds:
\begin{enumerate}

\item $\Var_{y \sim N(0,1)^n}[\tail_t(p_\ell(y))] \leq \eta$, ~~~~~~~~~~
or ~~~~~~~~~~
(b) ~~ ${\frac {(\lambda_{\max}(\tail_t(p_\ell))^2}{\Var[\tail_t(p_\ell)]}}
\leq \epsa$.

\end{enumerate}
\end{enumerate}

(Non-idealized lemma about {\bf Change-Basis}:)
This is the same as the idealized lemma
except that (1) above is replaced by
\begin{equation} \label{eq:pqclose}
\dk(p,q) \leq O(\epsa).
\end{equation}

\end{lemma}

\ifnum\confversion=1
In a real number model of computation, the ``Idealized Assumptions'' hold true and we obtain the ``idealized'' version of 
Lemma \ref{lem:change-basis}.  In a bit complexity model the ``Idealized Assumptions'' do not hold in that we cannot
exactly compute the desired quantities, and instead high-accuracy approximations must be used.
See the end of Section 2 (Preliminaries) of the full version for a discussion of approximation issues and error bounds.
\fi

\ifnum\confversion=0
\subsection{The {\bf Process-Polys} procedure.} \label{sec:processpolys}

In this subsection we describe and analyze the {\bf Process-Polys} procedure.
Our main result about this procedure is the following:

\begin{lemma} \label{lem:process-polys}
There is a deterministic procedure {\bf Process-Polys}
which runs in time $\poly(n, k, t, 1/\eps', 1/\eta)$ and has the following performance guarantee:
Given as input degree-2 polynomials $p_1(y),\dots,p_k(y)$ satisfying item (2)
of Lemma \ref{lem:change-basis}, an integer $0 \leq t \leq n$,
and a parameter $\eta$,
{\bf Process-Polys} outputs a vector $r=(r_1,\dots,r_k)$
of degree-2 polynomials over $\R^n$, and
a value $0 \leq k' \leq k$, such that $r,t,k'$ satisfy
the following properties:

\begin{enumerate}

\item ($r$ is as good as $p$ for the purpose of approximate counting:)
\[
\left|
\Pr_{y \sim N(0,1)^n}[\forall \ell \in [k],  r_\ell(y) \leq 0]
-
\Pr_{x \sim N(0,1)^n}[\forall \ell \in [k],  p_\ell(y) \leq 0]
\right| \leq O(\eps);
\]

\item For any restriction $\rho=(\rho_1,\dots,\rho_t) \in \R^n$ and
all $1 \leq \ell \leq k'$, the polynomial $r_\ell|_{\rho}$
has degree at most 1;

\ignore{
\item For all $1 \leq i \leq k'$, the polynomial
$\tail_t(r_i(y))$ consists only of a constant term;
}

\item For all $k' < \ell \leq k$, the polynomial $r_\ell(y)$ has
${\frac {\lambda_{\max}(\tail_t(r_\ell))^2}{\Var[\tail_t(r_\ell)]}}
\leq \epsa$;

\item For all $k' < \ell \leq k$, the polynomial $r_\ell(y)$ has
$\Var[\quadtail_t(r_\ell)] \geq \eta/2.$

\end{enumerate}

\end{lemma}

(Looking ahead, in Section \ref{sec:transform-proof}
Items (2)--(4) of Lemma \ref{lem:process-polys}
will be used to show that for most restrictions
$\rho=(\rho_1,\dots,\rho_t)$,
the distribution of $(r_1|\rho,\dots,r_k|\rho)$
is close to the distribution of a multivariate Gaussian with the right
mean and covariance.
Item (2) handles polynomials $r_1,\dots,r_{k'}$ and
Items (3) and (4) will let us use Theorem \ref{thm:small-dk} for the
remaining polynomials.)

\begin{framed}
\noindent {\bf Process-Polys}

\smallskip

\noindent {\bf Input:}  $k$-tuple $p=(p_1,\dots,p_k)$ of degree-2
polynomials $p_\ell(y_1,\dots,y_n)$ such that $\Var_{y \sim N(0,1)^n}
[p_\ell(y)]=1$; integer $t \geq 0$; parameter $\eta>0$.

\noindent {\bf Output:} $k$-tuple $r=(r_1,\dots,r_k)$ of degree-2
polynomials $r_\ell(y_1,\dots,y_n)$ and integer $0 \leq k' \leq k$.

\medskip

\begin{enumerate}

\item Reorder the polynomials $p_1(y),\dots,p_k(y)$ so that
$p_1,\dots,p_{k_1}$ are the ones that have $\Var[\tail_t(p_\ell)] \leq \eta$.
For each $\ell \in [k_1]$,
define $r_\ell(y) = \head_t(p_\ell(y)) + \E[\tail_t(p_\ell)]$.

\item Reorder the polynomials $p_{k_1+1}(y),\dots,p_{k}(y)$ so that
$p_{k_1+1}(y),\dots,p_{k_2}(y)$ are the ones that have
$\Var[\quadtail_t(p_{\ell}(y))] \leq \eta/2.$
For each $\ell \in [k_1+1,\dots,k_2]$,
define $r_\ell(y) = p_\ell(y) - \quadtail_t(p_\ell(y)).$

\item For each $\ell \in [k_2+1,\dots,k]$ define $r_\ell(y)=p_\ell(y).$
Set $k'=k_2$ and output $(r_1,\dots,r_k),k'.$
\end{enumerate}

\end{framed}

Recall that for $1 \leq \ell \leq k$ each polynomial $p_\ell$ is of the form
\[
p_\ell(y) = \sum_{1 \leq i \leq j \leq n}
\alpha^{(\ell)}_{ij} y_iy_j + \sum_{1 \leq i \leq n} \beta^{(\ell)}_i y_i
+ c^{(\ell)}.
\]
Because of Step 2 of {\bf Process-Polys},
for $1 \leq \ell \leq k_1$ we have that each polynomial
$r_{\ell}$ is of the form
\[
r_\ell(y) = \sum_{1 \leq i \leq t, j \geq i}
\alpha^{(\ell)}_{ij} y_iy_j + \sum_{1 \leq i \leq t} \beta^{(\ell)}_i y_i
+ c^{(\ell)},
\]
which gives part (2) of the lemma for $1 \leq \ell \leq k_1$.
Because of Step 3, for $k_1+1 \leq \ell \leq k_2$ we have that
each polynomial $r_{\ell}$ is of the form
\[
r_\ell(y) = \sum_{1 \leq i \leq t, j \geq i}
\alpha^{(\ell)}_{ij} y_iy_j + \sum_{1 \leq i \leq n} \beta^{(\ell)}_i y_i
+ c^{(\ell)},
\]
which gives part (2) of the lemma for $k_1+1 \leq \ell \leq k_2=k'.$
For $k_2 + 1 \leq \ell \leq k$ each polynomial $r_\ell(y)$
is of the form
\[
r_\ell(y) = \sum_{1 \leq i \leq j \leq n}
\alpha^{(\ell)}_{ij} y_iy_j + \sum_{1 \leq i \leq n} \beta^{(\ell)}_i y_i
+ c^{(\ell)}
\quad \quad \text{with~}
\Var[\quadtail_t(p_{\ell}(y))] > \eta/2,
\]
which gives part (4) of the lemma.

Part (3) of the lemma
follows immediately from item (2) of Lemma \ref{lem:change-basis}.
Thus the only part which remains to be shown is part (1).

We first deal with the polynomials $r_1,\dots,r_{k_1}$ using
the following simple claim:
\begin{claim}\label{clm:remove-low-var}
For each $\ell \in [k_1]$ we have that the $r_\ell(x)$ defined in Step 1
of {\bf Process-Polys} satisfies
$$\Pr_{x \sim N(0,1)^n}[\sign(r_\ell(x) \neq \sign(p_\ell(x))]
\le O(\sqrt{\log(1/\eta)} \cdot \eta^{1/4}).$$
\end{claim}
\begin{proof}
Recall that for $\ell \in [k_1]$ we have 
$r_\ell = \head_t(p_\ell) + \E[\tail_t(p_\ell)]$
while $p_\ell = \head_t(p_\ell) + \tail_t(p_\ell).$  Hence 
$\sign(r_\ell(x) \neq \sign(p_\ell(x))$ only if for some $s>0$ we have both
\[
|\head_t(p_\ell(x)) + \tail_t(p_\ell(x))| \leq s \sqrt{\eta}
\quad \text{and} \quad
|\tail_t(p_\ell(x)) - \E[\tail_t(p_\ell(x))]| > s \sqrt{\eta}.
\]

To bound the probability of the first event, 
recalling that $\dk(p_\ell,q_\ell) \leq O(\epsa)$ 
(by part (1) of Lemma \ref{lem:change-basis})
and that $\Var[q_\ell]=1$, it easily follows 
that $\Var[p_\ell]=\Theta(1)$.
Hence the Carbery-Wright inequality (Theorem \ref{thm:cw})
implies that
\begin{equation} \label{eq:parsnip}
\Pr_x[|\head_t(p_\ell(x)) + \tail_t(p_\ell(x))]|] \leq s \sqrt{\eta}
\leq O(s^{1/2} \eta^{1/4}).
\end{equation}

For the second event, we recall that $\Var[\tail_t(p_\ell)]\leq \eta$,
and hence for $s>e$ we may apply 
Theorem \ref{thm:dcb} to conclude that 
\begin{equation} \label{eq:radish}
\Pr_x[|\tail_t(p_\ell(x)) - \E[\tail_t(p_\ell(x))]| > s \sqrt{\eta}]
\leq O(e^{-s}).
\end{equation}
Choosing $s=\Theta(\log(1/\eta))$ we get that the RHS of (\ref{eq:parsnip})
and (\ref{eq:radish}) are both $\Theta(\sqrt{\log(1/\eta)} \cdot \eta^{1/4})$,
and the claim is proved.
\end{proof}

It remains to handle the polynomials $r_{k_1+1},\dots,r_{k_2}.$
For this we use the following claim:

\begin{claim}\label{clm:remove-small-quad-part}
For each $\ell \in [k_1+1,\dots,k_2]$
we have that the $r_\ell(x)$ defined in Step 2
of {\bf Process-Polys} satisfies
$$\Pr_{x \sim N(0,1)^n}[\sign(r_\ell(x) \neq \sign(p_\ell(x))]
\leq O(\sqrt{\log(1/\eta)} \cdot \eta^{1/4}).$$
\end{claim}
\begin{proof}
The proof is similar to Claim \ref{clm:remove-low-var}.
Recall that for $\ell \in [k_1+1,k_2]$ we have 
$r_\ell = p_\ell - \quadtail_t(p_\ell)$.  Hence
$\sign(r_\ell(x)) \neq \sign(p_\ell(x))$ only if for some $s>0$ we have both
\[
|p_\ell(x) + \E[\quadtail_t(p_\ell)]| \leq s \sqrt{\eta}
\quad \text{and} \quad
|\quadtail_t(p_\ell(x)) - \E[\quadtail_t(p_\ell)]| > s\sqrt{\eta}.
\]
For the first inequality, as above we have that
$\Var[p_\ell]=\Theta(1)$ so as above we get that
$\Pr_x[|p_\ell(x) + \E[\quadtail_t(p_\ell)]|  \leq s \sqrt{\eta}]
\leq O(s^{1/2}\eta^{1/4}).$  For the second
inequality we have $\Var[\quadtail_t(p_\ell)] \leq \eta/2$ so 
as above we get that 
$\Pr_x[|\quadtail_t(p_\ell) - \E[\quadtail_t(p_\ell)]| > s
\sqrt{\eta}] \leq O(e^{-s}).$  Choosing $s = \Theta(\log(1/\eta))$
as before the claim is proved.

\end{proof}

Recalling that $\eta = \Theta((\eps/k)^4/(\log(k/\eps))^2)$,
Claims \ref{clm:remove-low-var} and
\ref{clm:remove-small-quad-part}, together with a union bound,
give Lemma \ref{lem:process-polys}.
\qed

\subsection{Proof of Theorem \ref{thm:transform}.}
\label{sec:transform-proof}
Given what we have done so far in this section with the
{\bf Change-Basis} and {\bf Process-Polys} procedures,
the proof of Theorem \ref{thm:transform} is simple.
Item (1) of Lemma \ref{lem:change-basis} and Item (1) of Lemma
\ref{lem:process-polys} immediately give part (i) of Theorem
\ref{thm:transform}.  For part (ii),
\ignore{recall Lemma \ref{lem:process-polys} and }consider 
any restriction  $\rho=(\rho_1,\dots,\rho_t) \in \R^t$
fixing variables $y_1,\dots,y_t$ of the polynomials $r_1,\dots,r_k$.

We begin by observing that if the value $k'$ returned by {\bf Process-Polys}
equals $k$, then Item (2) of Lemma \ref{lem:process-polys}
ensures that for all $1 \leq \ell \leq k$ the restricted
polynomial $r_\ell|_\rho(y)$ has degree at most 1.
In this case the distribution of $(r_1|_\rho(y),\dots,r_k|_\rho(y))$ for
$y \sim N(0,1)^n$ is precisely that of a multivariate Gaussian over $\R^k$.  Since such
a multivariate Gaussian is completely determined by its mean and covariance matrix,
in this case we actually get that $\dk(r|_\rho,N(\mu(r|_\rho),\Sigma(r_\rho)))=0.$
So for the rest of the argument we may assume that $k' < k$, and consequently that
there is at least one polynomial $r_k$ that has 
${\frac {\lambda_{\max}(\tail_t(r_k))^2}{\Var[\tail_t(r_k)]}}
\leq \epsa$ and $\Var[\quadtail_t(r_k)] \geq \eta/2.$

First suppose that no restricted polynomial
$r_\ell|_\rho$ has $\Var[r_\ell|_\rho] > 1$.
Item (2) of Lemma \ref{lem:process-polys}
ensures that for $1 \leq \ell \leq k'$ the restricted
polynomial $r_\ell|_\rho(y)$ has degree 1 (note that in terms
of Theorem \ref{thm:small-dk}, this means
that the maximum magnitude of any eigenvalue of $r_\ell|_\rho$
is zero).  Now consider any $\ell \in [k'+1,k].$  Recalling 
Remark \ref{rem:restric}, we have that the polynomial $r_\ell|_\rho$
equals $\tail_t(r_\ell) + L$ for some affine form $L$.  Hence
\[
|\lambda_{\max}(r_\ell|_\rho)| =
|\lambda_{\max}(\tail_t(r_\ell))|
\leq \sqrt{{\Var[\tail_t(r_\ell)]} \cdot
\epsa} \leq O(\sqrt{\epsa}).
\]
where the first inequality is by Item (3).  The second inequality holds because
for $\ell \in [k'+1,k]$, the polynomial $r_\ell$ output by
{\bf Process-Polys} is simply $p_\ell$, so we have
$\Var[\tail_t(r_\ell)] \leq \Var[p_\ell] \leq \E[p_\ell^2].$
As in the proof of 
Claim \ref{clm:remove-low-var} we have that $\E[p_\ell^2] = O(1)$, giving the
second inequality above.
\ignore{Item (3) ensures
that for $k' < \ell \leq k$ the restricted polynomial $r_\ell|_\rho(y)$ has
${\frac {\lambda_{\max}(r_\ell|_\rho)^2}{\Var[r_\ell|_{\rho}]}} \leq
\epsa$, and hence
$|\lambda_{\max}(r_\ell|_\rho)| \leq \sqrt{{\Var[r_\ell|_{\rho}]} \cdot
\eps^2} \leq \eps$.}
Item (4) ensures that
that for $\ell  \in [k'+1,k]$ we have $\Var[r_\ell|_\rho(y)] =
\Var[\tail_t(r_\ell) + L] \geq
\Var[\quadtail_t(r_{\ell}(y))] \geq \eta/2.$
Thus we may apply Theorem \ref{thm:small-dk}
and conclude that the distribution of $(r_1|\rho,\dots,r_k|\rho)$
is $O(k^{2/3}\epsa^{1/12}/\eta^{1/6})$-close 
(i.e. $O(\eps)$-close) in $\dk$
to the distribution of the appropriate multivariate Gaussian, as claimed
in the theorem. 

Finally, consider the case that some restricted polynomial
$r_\ell|_\rho$ has $\Var[r_\ell|_\rho] > 1$.
In this case rescale each such restricted polynomial
$r_\ell|_\rho$ to reduce its variance down to 1; let
$\tilde{r}_1|_\rho,\dots,\tilde{r}_k|_\rho$ be the
restricted polynomials after this rescaling.
As above for $1 \leq \ell \leq k'$ we have that each restricted polynomial
$\tilde{r}_\ell|_\rho$ has $\lambda_{\max}(\tilde{r}_\ell|_\rho)=0$,
so consider any $\ell \in [k'+1,k].$
The rescaled polynomials $\tilde{r}_\ell$ satisfy
$\tilde{r}_\ell|_\rho = \tail_t(\tilde{r}_\ell) + \tilde{L}$, and we have
\[
{\frac {\lambda_{\max}(\tail_t(\tilde{r}_\ell))^2}
{\Var[\tail_t(\tilde{r}_\ell)]}} =
{\frac {\lambda_{\max}(\tail_t(r_\ell))^2}
{\Var[\tail_t(r_\ell)]}} \leq \epsa,
\]
so we get
\[
|\lambda_{\max}(\tilde{r}_\ell|_\rho)| = 
|\lambda_{\max}(\tail_t(\tilde{r}_\ell|_\rho))| \leq
\sqrt{\Var[\tail_t(\tilde{r}_\ell)] \cdot \epsa} \leq
\sqrt{\Var[\tail_t(r_\ell)] \cdot \epsa} \leq O(\sqrt{\epsa}),
\]
where for the penultimate inequality we recall that $\tilde{r}_\ell$ is
obtained by scaling $r_\ell$ down.  By assumption we have that
some $\ell$ has $\Var[\tilde{r}_\ell|_\rho]=1$, so we can 
apply Theorem \ref{thm:small-dk} and conclude that 
\[
\dk(\tilde{r}_\ell|_\rho,
N(\mu(\tilde{r}|_\rho),\Sigma(\tilde{r}|_\rho)) \leq O(k^{2/3} \epsa^{1/12}).
\]
Un-rescaling to return to $r_\ell$ from $\tilde{r}_\ell$, we get that
\[
\dk(r_\ell|_\rho,
N(\mu(r|_\rho),\Sigma(r|_\rho)) \leq O(k^{2/3} \epsa^{1/12}) = o(\eps),
\]
and Theorem \ref{thm:transform} is proved.
\qed

\fi

\section{Proof of Theorem \ref{thm:maingauss-informal}:
Efficient deterministic approximate counting using
transformed degree-2 Gaussian polynomials} \label{sec:gauss}
\ifnum\confversion=0
Throughout this section we focus on counting intersections of degree-$2$ PTFs.
The proof for an arbitrary $k$-junta follows by expressing it as a disjunction of $\mathrm{AND}_k$
functions and a union bound. \fi

Given Theorem \ref{thm:transform}, there is a natural approach for the counting algorithm {\bf Count-Gauss}, corresponding to the following steps:

\begin{framed}
\noindent {\bf Count-Gauss}

\smallskip

\noindent {\bf Input:}  $k$-tuple $p=(p_1,\dots,p_k)$ of degree-2
polynomials $p_\ell(y_1,\dots,y_n)$, $\ell \in [k]$, such that $\Var_{y \sim N(0,1)^n}
[p_\ell(y)]=1$; parameter $\eps>0$.

\noindent {\bf Output:} An $\pm O(\eps)$ additive approximation to the probability $\Pr_{x \sim N(0,1)^n} [\forall \ell \in [k], p_\ell(x) \ge 0].$

\medskip

\begin{enumerate}

\item Run {\bf Transform}$\left(p, \eps \right)$
to obtain a $k$-tuple of polynomials $r=(r_1,\dots,r_k)$ each of unit variance and a 
value $0 \leq t \leq O(k \ln(1/\eps)/\eps^2)$.  

\item Deterministically construct a product distribution $D^t = \otimes_{i=1}^t D_i$ supported
on a set $S \subseteq \R^t$ of cardinality ${(kt/\eps)^{O(t)}}$ such that a $t$-tuple
$\tau=(\tau_1,\dots,\tau_t) \in \R^t$ drawn from $D^t$ 
is ``close'' to a draw of $\rho=(\rho_1,\dots,\rho_t)$ from $N(0,1)^t$. In particular, $D_i = D$ for all $i \in [t]$,
where $D$ is a sufficiently accurate discrete approximation to $N(0,1).$ (See the proof of Lemma~\ref{lem:rounding}
for a precise description of the construction and guarantee.)

\item For each $\tau \in S$, 
simplify the polynomials $r_1,\dots,r_k$ by applying
the restriction to obtain $(r_1|_\tau,\dots,r_k|_\tau)$,
and compute the vector of means $\mu(r_\tau)$
and matrix of covariances $\Sigma(r_\tau)$.  

\item Finally, for each $\tau \in S$,
deterministically compute a $\pm \eps$-accurate additive 
approximation to the probability 
$\Pr_{y \sim N(\mu(r_\tau),\Sigma(r_\tau))} [ \forall i \in [k],  \ y_i \geq 0]$; let $p_\tau$ be the
value of the approximation that is computed.
Average all the values of $p_\tau$ obtained for each value $\tau \in S$ , and return the average.

\end{enumerate}
\end{framed}

Recall that the $k$-vector of polynomials $r = (r_1. \ldots, r_k)$ constructed in Step~1 satisfies
the statement of Theorem~\ref{thm:transform}. In particular, for every restriction of the first $t$ variables, the restricted
polynomials are $\eps$-close in Kolmogorov distance to a Gaussian with the corresponding mean and covariance matrix.
Hence, for each possible restriction $\rho$ of these $t$ variables, the probability that the restricted intersection of polynomials
is satisfied is $\eps$-close to the quantity $\Pr_{y \sim N(\mu(r_\rho),\Sigma(r_\rho))} [ \forall i \in [k],  \ y_i \geq 0]$.
Hence, if we could take ``all'' possible restrictions of these $t$ variables,  compute the corresponding probabilities and ``average'' the outcomes,
we would end up with an $\eps$-approximation to the desired probability. To achieve this efficiently, in Step~2, we construct
a sufficiently accurate discrete approximation to the normal distribution $N(0,1)^t$.



We have the following lemma:

\begin{lemma} \label{lem:rounding}
Let $r_{\ell}:\R^n \to \R$, $\ell \in [k]$,  be $k$ unit variance degree-$2$ polynomials. 
There exists a discrete distribution $D^t = \otimes_{i=1}^t D_i$ supported
on ${(kt/\eps)^{O(t)}}$ points that can be constructed explicitly
in output polynomial time such that 
$$ \left|
\Pr_{x \sim N^t(0,1), y \sim N^{n-t}(0,1)} \left[ \forall \ell \in [k],   r_{\ell}(x, y) \ge 0 \right] -  
\Pr_{\tilde{x} \sim D^t, y \sim N^{n-t}(0,1)} \left[ \forall \ell \in [k],  r_{\ell}(\tilde{x}, y) \ge 0 \right] 
\right| \leq O(\eps).$$
\end{lemma}
\ifnum\confversion=0
\begin{proof}
Before we proceed with the formal proof, we provide some intuition. The main technical point is how ``fine''
a discretization we need to guarantee an $\pm \eps$ approximation to the desired probability \[\Pr_{z \sim N^{n}(0,1)}[ \forall \ell \in [k],  r_{\ell}(z) \ge 0 ].\] 
Each component $D_j$, $j \in [t]$, of the product distribution $D^t$ will be a discrete approximation to the standard Gaussian distribution $N(0,1)$.
Consider a sample $x = (x_1, \ldots, x_t) \sim N^t(0,1)$ drawn from the standard Gaussian and its coordinate-wise closest discretized value
$\tilde{x} = (\tilde{x}_1, \ldots, \tilde{x}_t)$. The main idea is to construct each $D_j$ in such a way so that with probability at least $1-{O(\eps/k)}$ over $x$, 
the absolute difference $\max_{j \in [t]} |x_j - \tilde{x}_j|$ is at most $\delta$ (where $\delta$ is a sufficiently small quantity). Conditioning on this event, the difference between
the two probabilities $ \Pr_{x \sim N^t(0,1), y \sim N^{n-t}(0,1)} \left[ \forall \ell \in [k], r_{\ell}(x, y) \ge 0 \right]$ and 
$ \Pr_{\tilde{x} \sim D^t, y \sim N^{n-t}(0,1)} \left[ \forall \ell \in [k], r_{\ell}(\tilde{x}, y) \ge 0 \right]$
can be bounded from above by the probability of the following event: 
there exists $\ell \in [k]$ such that the polynomial $r_{\ell}(x, y)$ is ``close'' to $0$ or the difference between the two restricted polynomials 
$r_{\ell}(x, y) - r_{\ell}(\tilde{x}, y)$ is ``large''.
Each of these events can in turn be bounded by a combination of anti-concentration and concentration for degree-$2$ polynomials which completes the proof
by a union bound.



\medskip

\noindent {\bf Construction of the discrete distribution $D^t$.} The distribution $D^t = \otimes_{j=1}^t D_j$ is a product distribution, whose individual marginals
$D_j$, $j \in [t]$, are identical, i.e., $D_j = D$. The distribution $D$ is a discrete approximation to $N(0,1)$. Intuitively, to construct $D$ we proceed as follows.
After truncating the ``tails'' of the Gaussian distribution, we partition the domain into a set of subintervals $I_i$. The distribution $D$ will be supported
on the leftmost points of the $I_i$'s and the probability mass of each such point will be approximately equal to the mass the Gaussian distribution 
assigns to the corresponding interval. More specifically, let us denote {$\eps' = \eps/(kt)$ and $M = \Theta (\sqrt{\log(1/\eps')})$}.
Then $D$ is supported on the grid of points $s_i = i \cdot \delta$, where $i$ is an integer and $\delta$ is chosen (with foresight) to be 
{$\delta \eqdef \Theta \left( \eps^2 / (k^2 \log(k/\eps)) \right)$}.  The range of the index $i$ is such that $|i| \cdot \delta \le M$, i.e. $i \in [-s, s]$, where
$s \in \Z_+$ with $s = O((1/\delta)\cdot M)$.

The probability mass that $D$ 
assigns to the point $s_i = i \cdot \delta$ is approximately equal to the probability that a standard Gaussian random variable assigns to the interval $I_i = [s_i, s_{i+1})$.
In particular, if $\Phi(I)$ denotes the probability that a standard Gaussian puts in interval $I$, we will guarantee that 
\begin{equation} \label{eqn:close}
\littlesum_{i} \left|  \Phi(I_i) - D(s_i) \right| \le \eps'.
\end{equation}
To achieve this we make the error in each interval to be at most $\eps'$ divided by the number of intervals.
It is clear that $D$ can be constructed explicitly in time $\poly(tk/\eps)$.
{Note that, as a consequence of (\ref{eqn:close}) we have that $\dk(D, N(0,1)) \le \eps'$.}

\medskip

\noindent {\bf Properties of $D$.} We define the natural coupling between $N(0,1)$ and $D_j$, $j \in [t]$: a sample $x_j \sim N(0,1)$ such that $x_j \in I_i$ is coupled to the point $\tilde{x}_j$ that corresponds to the left endpoint of the interval $I_i$. If $x_j$ is such that $|x_j| > M$ we map $x_j$ to an arbitrary point.
This defines a coupling between the product distributions $D^t$ and $N^t(0,1)$.
The main property of this coupling is the following: 

\begin{fact} \label{fact:couple}
With probability at least $1-{O(\eps/k)}$ over a sample $x \sim N^{t}(0,1)$ its ``coupled'' version $\tilde{x}$ satisfies ${\max_{j \in [t]}} |x_j - \tilde{x}_j| \le \delta$. 
\end{fact}
\begin{proof}
For each coordinate $j \in [t]$, it follows from Condition (\ref{eqn:close}) and the concentration of the standard Gaussian random variable that
with probability at least $1-\eps'$ we have $|x_j - \tilde{x}_j| \le \delta$. The fact then follows by a union bound.
\end{proof}

We henceforth condition on this event. For technical reasons, we will further condition on the event that {$ \eps' \le |x_j| \le M$ for all $j \in [t]$}. 
This event will happen with probability at least $1-{O(\eps/k)}$, by Gaussian concentration and anti-concentration followed by a union bound. 
Note that the complementary event affects the desired probabilities by at most $\eps$.


Fix an $x = (x_1, \ldots, x_t)$ with $ \eps' \le |x_j| \le M$ for all $j \in [t]$
and a value $\tilde{x} = (\tilde{x}_1, \ldots, \tilde{x}_t)$ such that $\max_{j \in [t]} |x_j - \tilde{x}_j| \le \delta$.
For $\ell \in [k]$, consider the difference $e_{\ell}(x, \tilde{x}, y) = r_{\ell}(x, y) - r_{\ell}(\tilde{x}, y)$ as a random variable in $y \sim N(0,1)^{n-t}$.
We have the following claim:

\begin{claim} \label{claim:var-upper}
We have that $\Var_y [e_{\ell}] = O(\delta^2)$.
\end{claim}
\begin{proof}
Let $r_{\ell}(x_1, \ldots, x_n) = \littlesum_{i,j} a_{ij}x_ix_j + \littlesum_i b_i x_i + C$.
By our assumption that $\Var[r_{\ell}]=1$ and Claim~\ref{clm:variance-bound-SS}, 
it follows that the sum of the squares of the coefficients of $r_{\ell}$ is in $[1/2, 1]$.
A simple calculation yields that the difference between $r_{\ell}(x_1, x_2, \ldots, x_n)$ and $r_{\ell}(\tilde{x}_1, \ldots, \tilde{x}_t, x_{t+1}, \ldots, x_n)$
is at most 
\[  \littlesum_{1\le i \le j \le t} a_{ij} (x_ix_j - \tilde{x}_i \tilde{x}_j) + \littlesum_{i \le t, j \ge t+1} a_{ij} (x_i - \tilde{x}_i) x_j + \littlesum_{i\le t} b_i (x_i - \tilde{x}_i)   \]
Taking into consideration our assumption that the sum of the squared coefficients of $r_{\ell}$ is at most $1$ 
and that $|x_j - \tilde{x}_j| \le \delta$ for all $j \in [t]$, the variance of the above quantity term can be bounded from above by $O(\delta^2)$.
\end{proof}

\ignore{ 
 
We will also need the following claim bounding from below the variance of $r(x, y)$:

\begin{claim} \label{claim:var-lower}
We have that $\Var_y [r(x, y)] = \Omega({\eps'^2 \cdot \eps^4 / \log^2(1/\eps)})$.
\end{claim} 
\begin{proof}
Recall that we have conditioned on the fact that {$|x_j|>\eps'$ for all $j \in [t]$}.
By expanding $r(x_1, x_2, \ldots, x_n)$ as a function of $x_{t+1}, \ldots, x_n$ it follows that its sum of squared coefficients
is 
\[\littlesum_{i \le t, j \ge t+1} a_{ij}^2x_i^2 + \littlesum_{t+1 \le i \le j \le n} a_{ij}^2 + \littlesum_{i \ge t+1} b_i^2 \] 
{We now recall that by assumption the above sum of squares is at least   $\eps^4 / \log^2(1/\eps)$.
Indeed, if this is not the case, the polynomial $r$ is $O(\eps)$-close in Kolmogorov distance to the $t$-junta 
$\littlesum_{1 \le i \le j \le t} a_{ij}x_ix_j + \littlesum_{i \le t} b_i x_i + C.$
Under this assumption, it follows that the variance $\Var_y [r(x, y)]$ is at least $\Omega ( \eps'^2 \cdot \eps^4 / \log^2(1/\eps)).$
}
\end{proof}
 
} 
 
Given a value of $\gamma > 0$, the two desired probabilities differ only if there exists $\ell \in [k]$ such that
\begin{equation} \label{eqn:conc}
\Pr_{x, \tilde{x}, y} [|e_{\ell}(x, \tilde{x}, y)| \ge \gamma]
\end{equation} 
or
\begin{equation} \label{eqn:ac} 
\Pr_{x, y} [|r_{\ell}(x, y)| \le \gamma].
\end{equation}
We will select the parameter $\gamma$ appropriately so that
for a given $\ell \in [k]$, both probabilities above are at most $O(\eps/k)$. 
The proof of the lemma will then follow by a union bound {over $\ell$}.

For fixed $x, \tilde{x}$, an application of the Chernoff bound (Theorem~\ref{thm:dcb}) in conjunction with Claim~\ref{claim:var-upper}
implies that $\Pr_{y} [|e_{\ell}(x, \tilde{x}, y)| \ge \gamma]$
is at most {$\tilde{\eps} = \eps/k$} as long as $\gamma = \Omega(\log(1/\tilde{\eps}) \delta)$.
By Fact~\ref{fact:couple} it thus follows that (\ref{eqn:conc}) is at most $O(\eps/k)$.
Similarly, since $\Var[r_{\ell}]=1$, by choosing $\gamma = \Theta(\tilde{\eps}^2)$, Carbery--Wright (Theorem~\ref{thm:cw})  implies 
that (\ref{eqn:ac}) is at most {$O(\tilde{\eps}).$} By our choice of $\delta$, it follows that for this choice of $\gamma$ we indeed have that $\gamma = \Omega(\log(1/\tilde{\eps}) \delta)$, which completes the proof. 
\end{proof}
\fi

For Step~4 we note that the corresponding problem is that of counting an intersection 
of $k$ halfspaces with respect to a Gaussian distribution over $\R^k$. We recall that, by Theorem~1.5 of~\cite{GOWZ10},
$s = \tilde{O}(k^6/\eps^2)$-wise independence $\eps$-fools such functions. Since we are dealing with a $k$-dimensional problem,
any explicit construction of an $s$-wise independent distribution yields a deterministic $\eps$-approximate counting algorithm
that runs in time $k^{O(s)}$, completing the proof of Theorem~\ref{thm:maingauss-informal}.

\section{Deterministic approximate counting for $g(\sign(q_1(x)),\dots,
\sign(q_k(x)))$ over $\{-1,1\}^n$}
\label{sec:Bool}

In this section we extend the deterministic approximate counting result that
we established for the Gaussian distribution on $\R^n$ to the uniform
distribution over $\{-1,1\}^n$, and prove Theorem~\ref{thm:main-boolean-informal}.
As discussed in the introduction,
there are three main ingredients in the
proof of Theorem~\ref{thm:main-boolean-informal}.
The first, of course, is the Gaussian counting result,
Theorem \ref{thm:maingauss-informal}, established earlier.
The second is a deterministic algorithmic
regularity lemma for $k$-tuples of low-degree polynomials:

\begin{lemma} \label{lem:k-reg-alg} [algorithmic
regularity lemma, general $k$, general $d$]
There is an algorithm $\tt{ConstructTree}$ with the following property:

Let $p_1,\dots,p_k$ be degree-$d$ multilinear polynomials
with $b$-bit integer coefficients
over $\{-1,1\}^n$.  Fix $0 < \tau, \eps, \delta < 1/4$.
Algorithm $\tt{ConstructTree}$ (which is deterministic) runs in time
$\poly(n,b,2^{D_{d,k}(\tau,\eps,\delta)})$ and
outputs a decision tree
$T$ of depth at most
\ifnum\confversion=1
$D_{d,k}(\tau,\eps,\delta) := 
\left( {\frac 1 \tau} \cdot \log {\frac 1 \eps} \cdot
\right)^{(2d)^{\Theta(k)}} \cdot
\log {\frac 1 \delta} .
$
\fi
\ifnum\confversion=0
\[
D_{d,k}(\tau,\eps,\delta) := 
\left( {\frac 1 \tau} \cdot \log {\frac 1 \eps} \cdot
\right)^{(2d)^{\Theta(k)}} \cdot
\log {\frac 1 \delta} .
\]
\fi
Each internal node of the tree is labeled with a variable and each leaf 
$\rho$ is
labeled with a $k$-tuple of polynomials $((p_1)_\rho,\dots,(p_k)_\rho)$ 
and with a
$k$-tuple of labels $(\lab_1(\rho),\dots,\lab_k(\rho)).$
For each leaf $\rho$ and each
$i \in [k]$ the polynomial $(p_i)_\rho$ is the polynomial obtained by
applying restriction $\rho$ to polynomial $p_i$, and $\lab_i(\rho)$
belongs to the set $\{+1,-1,$``fail'', ``regular''$\}.$
The tree $T$ has the following properties:

\begin{enumerate}

\item For each leaf $\rho$ and index $i \in [k]$, if $\lab_i(\rho) \in \{+1,-1\}$,
then $\Pr_{x \in \{-1,1\}^n}[\sign((p_i)_\rho(x)) \neq \lab_i(\rho)] \leq \eps$;

\item For each leaf $\rho$ and index $i \in [k]$, if $\lab_i(\rho) = $``regular''
then $(p_i)_\rho$ is $\tau$-regular; and

\item With probability at least $1-\delta$, a random path from the root reaches
a leaf $\rho$ such that $\lab_i(\rho) \neq$``fail'' for all $i \in [k].$

\end{enumerate}

\end{lemma}

The third ingredient is the following version of the
multidimensional invariance principle, which lets us move
from the Gaussian to the Boolean domain:

\begin{theorem} \label{thm:our-invar}
Let $p_1(x),\dots,p_k(x)$ be degree-$d$ multilinear
polynomials over $\{-1,1\}^n$, and let $P_i(x) = \sign(p_i(x))$ for $i=1,\dots,k$.
Suppose that each $p_i$ is $\tau$-regular.  Then
for any $g: \{-1,1\}^k \to \{-1,1\}$, we have that
\[
\left|
\Pr_{x \sim \{-1,1\}^n}[g(P_1(x),\dots,P_k(x)) = 1] -
\Pr_{{\cal G} \sim N(0,1)^n}[g(P_1({\cal G}),\dots,P_k({\cal G})) = 1]
\right|
\leq \widetilde{\eps} (d, \tau, k),
\]
where $\widetilde{\eps} (d, \tau, k) : = 2^{O(k)} \cdot 2^{O(d)} \cdot \tau^{1/(8d)}$.
\end{theorem}

The regularity lemma for $k$-tuples of polynomials, Lemma
\ref{lem:k-reg-alg}, requires significant technical work; 
we prove it in Section \ifnum\confversion=0\ref{sec:regularity}\else 7 of the full version\fi.
In contrast, Theorem \ref{thm:our-invar} is a fairly direct consequence of the
multidimensional invariance principle of Mossel \cite{Mos08} \ifnum\confversion=1 (see full version)\fi. 
\ifnum\confversion=0 We explain how Theorem \ref{thm:our-invar} follows from \cite{Mos08} 
in Section \ref{sec:our-invar}.
Before establishing the regularity lemma and the invariance principle that
we will use, though, we first show how Theorem \ref{thm:main-boolean-informal} follows
from these results.\fi
\ifnum\confversion=1
In the full version we show how these results can be combined to prove Theorem \ref{thm:main-boolean-informal}.
\fi

\ifnum\confversion=0

\medskip

\noindent {\bf Proof of Theorem \ref{thm:main-boolean-informal} using
Theorem \ref{thm:maingauss-informal}, Lemma \ref{lem:k-reg-alg} and
Theorem \ref{thm:our-invar}:}
The algorithm for approximating $\Pr_{x \sim \{-1,1\}^n}
[g(Q_1(x),\dots,Q_k(x))=1]$
to within an additive $\pm \eps$ works as follows.  It
first runs algorithm $\tt{ConstructTree}$ from Lemma \ref{lem:k-reg-alg}
with parameters $d,k,\tau_0,\eps_0,$ and $\delta_0$, where
$\tau_0$ satisfies $\\widetilde{\eps}(d,\tau_0,k) \leq \eps/4$,
$\eps_0$ equals $\eps/(4k)$, and
$\delta_0$ equals $\eps /4$,
to construct the decision tree $T$.  It initializes the value $\tilde{v}$
to be 0, and then iterates over all leaves $\rho$ of the tree $T$, adding
a contribution $\tilde{v}_\rho$ to $\tilde{v}$
at each leaf $\rho$ according to the following rules:
for a given leaf $\rho$ at depth $d_\rho$,

\begin{itemize}

\item If any $i \in [k]$ has $\lab_i(\rho) =$ ``fail'' then
the contribution $\tilde{v}_\rho$ from that leaf is 0.  Otherwise,

\item Let $\kappa(\rho)$ be the restriction of variables
$y_1,\dots,y_k$ corresponding to the string
$(\lab_1(\rho),\dots,\lab_k(\rho)) \in \{+1,-1,$ ``regular''$\}$,
so $\kappa(\rho)$ fixes variable $y_i$ to $b \in \{+1,-1\}$ if
$\lab_i(\rho) = b$ and $\kappa(\rho)$ leaves variable $y_i$
unfixed if $\lab_i(\rho)=$``regular.''
Run the algorithm of Theorem \ref{thm:maingauss-informal}, providing as input
the $k$-tuple of polynomials $((p_1)_\rho,\dots,(p_k)_\rho)$,
the Boolean function $g_{\kappa(\rho)}$ (i.e. $g$ with restriction
$\kappa(\rho)$ applied to it), and the accuracy parameter $\eps/4$;
let $\tilde{w}_\rho$ be the value thus obtained.  The contribution
from this leaf is
$\tilde{v}_\rho := \tilde{w}_\rho \cdot 2^{-d_\rho}$.

\end{itemize}

Theorem \ref{thm:maingauss-informal} and Lemma \ref{lem:k-reg-alg}
imply that the running time is as claimed; we now prove correctness.
Let $v$ denote the true value of $\Pr_{x \sim \{-1,1\}^n}[
g(Q_1(x),\dots,Q_k(x))=1].$  We may write
$v$ as $\sum_{\rho} v_\rho$, where the sum is over all leaves
$\rho$ of $T$ and $v_\rho = w_\rho \cdot 2^{-d_\rho}$
where
\[
w_\rho = \Pr_{x \sim \{-1,1\}^n}[g((Q_1)_\rho(x),\dots,
(Q_k)_\rho(x))=1].
\]
We show that $|v-\tilde{v}| \leq \eps$ by showing that
$\sum_{\rho} |\tilde{v}_\rho - v_\rho| \leq \eps.$
To do this, let us partition the set of all leaves $\rho$ of $T$
into two disjoint subsets $A$ and $B$, where a leaf $\rho$ belongs
to $A$ if some $i \in [k]$ has $\lab_i(\rho)=$``fail''.
Part (3) of Lemma \ref{lem:k-reg-alg} implies that
$\sum_{\rho \in A} 2^{-d_\rho} \leq \delta_0 = \eps/4$, so we have
that
\[
\sum_{\rho \in A} |\tilde{v}_\rho - v_\rho| =
\sum_{\rho \in A} v_\rho \leq
\sum_{\rho \in A} 2^{-d_\rho} \leq \eps/4.
\]
We bound $\sum_{\rho \in B} |\tilde{v}_\rho - v_\rho| \leq 3 \eps/4$
by showing that each leaf $\rho \in B$ satisfies
$|w_\rho - \tilde{w}_\rho| \leq 3 \eps/4$; this is sufficient since
\[
\sum_{\rho \in B} |\tilde{v}_\rho - v_\rho| =
\sum_{\rho \in B} 2^{-d_\rho}|\tilde{w}_\rho - w_\rho| \leq
\left(\max_{\rho \in B} |w_\rho - \tilde{w}_\rho|\right)
\cdot
\sum_{\rho \in B} 2^{-d_\rho} \leq
\max_{\rho \in B} |w_\rho - \tilde{w}_\rho|
\leq 3 \eps/4.
\]
So fix any leaf $\rho \in B$.  Let $S_{\kappa(\rho)} \subseteq [k]$
be the subset of those indices $i$ such that $\lab_i(\rho)=$
``regular''.  By part (2) of Lemma \ref{lem:k-reg-alg} we have that
$(p_i)_\rho$ is $\tau_0$-regular for each $i \in S_{\kappa(\rho)}.$
Hence we may apply Theorem \ref{thm:our-invar} to
the Boolean function $g_{\kappa(\rho)}:
\{-1,1\}^{S_{\kappa(\rho)}} \to \{-1,1\}$, and we get that
\begin{eqnarray}
&&
\left|
\Pr_{x \sim \{-1,1\}^n}[g_{\kappa(\rho)}((Q_1)_\rho(x),\dots,(Q_k)_{\rho}
(x)) = 1] -
\Pr_{{\cal G} \sim N(0,1)^n}[g_{\kappa(\rho)}
((Q_1)_\rho({\cal G}),\dots,(Q_k)_{\rho}({\cal G})) = 1]
\right| \nonumber \\
&\leq&
\widetilde{\eps}(d,\tau_0,k) \leq \eps/4. \label{eq:bool-to-gauss}
\end{eqnarray}

By Theorem \ref{thm:maingauss-informal} we have that
\begin{eqnarray}
\left|\tilde{w}_\rho -
\Pr_{{\cal G} \sim N(0,1)^n}[g_{\kappa(\rho)}
((Q_1)_\rho({\cal G}),\dots,(Q_k)_{\rho}({\cal G})) = 1]
\right|
\leq \eps/4.
\label{eq:alg-to-gauss}
\end{eqnarray}
Finally, part (1) of Lemma \ref{lem:k-reg-alg} and a union bound give that
\begin{eqnarray}
\left|
w_\rho -
\Pr_{x \sim \{-1,1\}^n}[g_{\kappa(\rho)}((Q_1)_\rho(x),\dots,(Q_k)_{\rho}
(x)) = 1]\right|
&\leq&
\sum_{i \in ([k] \setminus S_{\kappa(\rho)})}
\Pr_{x \sim \{-1,1\}^n}[(Q_i)_\rho(x) \neq \lab_i(\rho)] \nonumber\\
&\leq& k \cdot \eps_0 = \eps/4.
\label{eq:bool-restric-to-bool}
\end{eqnarray}
Combining (\ref{eq:bool-to-gauss}), (\ref{eq:alg-to-gauss}) and
(\ref{eq:bool-restric-to-bool}) with the triangle inequality we get that
$|w_\rho - \tilde{w}_\rho| \leq 3\eps/4$, which concludes the proof of
Theorem \ref{thm:main-boolean-informal}.
\qed

\subsection{Proof of Theorem \ref{thm:our-invar}}\label{sec:our-invar}.
We start by proving the theorem for the case that the 
$k$-junta $g$ is the AND$_{k}$ function.
In fact, in this particular 
case the dependence of the error on the parameter $k$ is polynomial.
The generalization to an arbitrary $k$-junta follows using a 
union bound and
the fact that any $k$-junta can be written as an OR of at most $2^k$ AND$_{k}$
functions, each of which is satisfied by a different point in $\{-1,1\}^k.$

The proof has two steps: In the first step we prove the theorem for ``smooth'' functions; in the second step 
we use FT-mollification to reduce the theorem to the smooth case.
The first step is an immediate application of Theorem~4.1 in~\cite{Mossel10}. 
In particular, the following statement is a corollary of his statement to our setting:

\begin{theorem}[\cite{Mossel10}, Corollary of Theorem~4.1] \label{thm:invariance-smooth}
Let $p_1(x), p_2(x), \ldots, p_k(x)$ be degree-$d$ multilinear polynomials (where either $x \in \{-1, 1\}^n$ or $x \in \R^n$) such that $\Var[p_i]=1$  
and $\max_j \Inf_j(p_i) \le \tau$ for all $i=1, \ldots, k$. Let $\Psi:\R^k \to \R$ be a $C^3$ function with 
$\|\Psi^{(\mathbf{i})}\|_{\infty} \le B$ 
for every vector $\mathbf{i} \in (\Z_{\geq 0})^n$ with 
$\|\mathbf{i}\|_1 \le 3$,
where $\Psi^{(\mathbf{i})}$ denotes the $\mathbf{i}$-th iterated
partial derivative of $\Psi$.  Then,
\[ \left| \E_{x \sim  \{-1, 1\}^n} 
\left[ \Psi \left(p_1(x),\ldots, p_k(x)\right) \right] -  
\E_{G \sim  N(0,1)^n} \left[ \Psi \left(p_1(G),\ldots, p_k(G)\right) \right]  
 \right| \le 
\eps: =  2Bk^{9/2}(8\sqrt{2})^d \cdot d \sqrt{\tau}. \]
\end{theorem}

\begin{remark}
{\em We now briefly explain how the above is obtained from Theorem~4.1 of~\cite{Mossel10}. 
Theorem~4.1 considers a $k$-dimensional multi-linear polynomial $q = (q_1, \ldots, q_k)$.
The variance of the $k$-vector $q$ is defined to be the sum of the variances of the individual components, i.e., $\Var[q] = \sum_{j \in [k]} \Var[q_j]$. Similarly,
the influence of the $i$-th variable on $q$ is defined as the sum of the influences of the components, i.e., $\Inf_i[q] = \sum_{j \in [k]} \Inf_i[q_j]$. 
The degree of $q$ is the maximum of the degree of the components. Note that when we apply Theorem~4.1 to our setting, the corresponding $k$-dimensional 
multi-linear polynomial $p = (p_1, \ldots, p_k)$ has variance equal to $k$. Similarly, the influence of each variable in $p$ is at most $k\tau$. 
Finally, the value $\alpha$ in the notation of~\cite{Mossel10} is by definition equal to $1/2$. (See the derivation on top of p. 21 of the ArXiV 
version of \cite{Mossel10}}.)
\end{remark}
\noindent Note that in Theorem \ref{thm:invariance-smooth} 
the error parameter $\eps$ depends polynomially on $k$ and exponentially on 
$d$. As we now show, when the $k$-junta $g$ is the AND function, the second 
step (FT-mollifcation) also results in a polynomial dependence on $k$.


Let $g$ be the AND function on $k$ variables. We assume (wlog) that the range of $g$ is $\{0, 1\}$ as opposed to $\{-1, 1\}$. 
Let $p = (p_1, \ldots, p_k)$ be our $k$-vector of degree-$d$ multilinear polynomials satisfying the assumptions of Theorem~\ref{thm:invariance-smooth}.
Denote by $\theta_i$ and $p'_i$ the constant and non-constant parts of $p_i$ respectively, for $i=1, \ldots, k$, so
$p_i(x) = p'_i(x) + \theta_i$ for $i=1, \ldots, k$, where $p'_i(x)$ is a degree-$d$ polynomial with constant term $0$ and variance $1$.

Consider the region $R = \{y_i + \theta_i \ge 0, i \in [k]\} \subseteq \R^k$.
We claim that, in order to prove  Theorem~\ref{thm:our-invar} for $g$ being 
the AND$_{k}$ function, it suffices to establish the existence of a smooth 
function $\Psi$ such that the following two bounds hold:
\begin{equation} \label{eqn:ft-goal}
\E_{x \sim {\cal D}} \left[ \Psi \left(p'_1(x),\ldots, p'_k(x)\right) 
\right]  \approx_{\delta}  \E_{x \sim {\cal D}}
\left[ I_R \left(p'_1(x),\ldots, p'_k(x)\right) \right],
\end{equation}
where ${\cal D}$ is taken either to be the uniform distribution over
$\bn$ or to be $N(0,1)^n$,
for an appropriately small value of $\delta$.
Indeed, given these two versions of Equation~\ref{eqn:ft-goal}, 
Theorem~\ref{thm:our-invar} follows from Theorem~\ref{thm:invariance-smooth} 
and the triangle inequality with $\widetilde{\eps}= 2\delta+\eps$.

To establish the existence of a smooth approximation $\Psi$ to $I_R$ satisfying ~\ref{eqn:ft-goal}, we appeal to Theorem~\ref{thm:smooth-apx}.
In particular, the smooth function $\Psi$ will be the function $\tilde{I}_c$ of that theorem, for an appropriately large value of the parameter $c>0$.
Note that there is a tradeoff between the relevant parameters: On the one hand, the higher the value of $c$, the better an approximation 
$\tilde{I}_c$ will be to $I_R$, and hence the smaller the parameter $\delta$ will be. On the other hand, when $c$ increases, so does the upper bound on the magnitude 
of the derivatives of $\tilde{I}_c$ (see the first condition of Theorem~\ref{thm:smooth-apx}). 
This in turn places a lower bound on the value of $B$ (the maximum value of the third derivative) in Theorem~\ref{thm:invariance-smooth} -- hence, the parameter
$\eps$ increases. As a consequence of this tradeoff, one needs to select the parameter $c$ carefully to minimize the total error 
of $\widetilde{\eps}= O(\delta+\eps)$.

We will additionally need to use the fact that the random vector $p' = (p'_1, \ldots, p'_k)$ is sufficiently anti-concentrated (so that the contribution to the error from the region
where $I_R$ and its FT-mollified version differ by a lot is sufficiently small). For the case of the Gaussian distribution, this follows immediately from the Carbery-Wright inequality (Theorem~\ref{thm:cw}).
For the case of the uniform distribution over the cube, this follows (as usual), by a combination of the ``basic'' invariance principle of~\cite{MOO10} combined with Theorem~\ref{thm:cw}.

We perform the calculation for the regular boolean case below. It turns out that this is the bottleneck quantitatively -- and it subsumes the 
Gaussian case (since the corresponding anti-concentration bound holds for the Gaussian case as well). We start by recording the following fact,
which is a corollary of \cite{MOO10} combined with Theorem~\ref{thm:cw}:
\begin{fact} \label{fact:regular-ac}
Let $q:\bn \to \R$ be a $\tau$-regular degree-$d$ polynomial with $\Var[q]=1$ and $\rho>0$. Then, for all $\theta \in \R$ we have
\[ \Pr_{x \in \bn} \left[ |p(x)-\theta| \le \rho \right] \le O(d\tau^{1/(8d)}) + O(d\rho^{1/d}).\]
\end{fact}

\noindent {\bf Choice of Parameters:} We set $\rho \eqdef O(\tau^{1/8})$ and choose the parameter $c$ in Theorem~\ref{thm:smooth-apx} equal to $c \eqdef k/\rho$. 
We proceed to bound from above the quantity 
\[  \left| \E_{x  \sim \{-1,1\}^n} \left[ 
I_R \left(p'_1(x),\ldots, p'_k(x)\right) \right]  -  
\E_{x \sim \{-1,1\}^n} \left[ \tilde{I}_c \left(p'_1(x),\ldots, 
p'_k(x)\right) \right]  \right|.\] 

We start by observing that for any $y \in \R^k$, the Euclidean distance
$\|y - \partial R\|$ is at least $\min_i |y_i + \theta_i|$. Hence 
by a union bound combined with the above fact we obtain
\[ \Pr_x [ \|p'(x) - \partial R\| \le \rho  ]  
\le \Pr_x [ \min_i \{|p'_i(x) + \theta_i|\} \le \rho] \le \littlesum_{i=1}^k  \Pr_x [ |p'_i(x) + \theta_i| \le \rho]  = O(kd\tau^{1/(8d)}).\]
Similarly, for $w \ge \rho$ we have 
\[ \Pr_x [ \|p'(x) -  \partial R\| \le w  ]  = O(kdw^{1/d}).\]
Using these inequalities and Theorem~\ref{thm:smooth-apx} we bound from above the desired quantity as follows:
\begin{eqnarray*}
&& \left| \E_{x} \left[ I_R \left(p'(x)\right) \right]  -  \E_{x} \left[ \tilde{I}_c \left(p'(x)\right) \right]  \right| \\
&\le& \E_x \left[ \left| I_R (p'(x))  - \tilde{I}_c (p'(x))  \right| \right]  \\
&\le& \Pr_x [ \|p'(x) -  \partial R\| \le \rho  ] + 
\littlesum_{s=0}^{\infty} \left( \frac{k^2}{c^2 2^{2s} \rho^2}  \right)  \Pr_x [ \|p'(x) - \partial R\| \le 2^{s+1}\rho]\\
&\le&  
O(kd\tau^{1/(8d)}) +O(kd \rho^{1/d}) \littlesum_{s=0}^{\infty} 2^{-2s}  2^{s/d}
\quad \text{(by our choice of $c=k/\rho$)}\\
&=&  O(kd\tau^{1/(8d)}).
\end{eqnarray*}
Hence we obtain Equation~\ref{eqn:ft-goal} 
for $\delta =  O(kd\tau^{1/(8d)}).$
\ignore{

So, it suffices for us to calculate how large $c$ needs to be in Theorem~12 so that the two expectations differ by $\delta$.
The calculation we need is essentially the same as the first part of the proof in Appendix~H of~\cite{DKNfocs10}. The only difference is that instead of
using the anti-concentration for linear forms, we need to use the corresponding quantitative expression for degree-$d$ polynomials. 
In particular, we need to calculate the right value of $\rho$ in their notation.
Basically, the total error term in our case will be at most 
\[ O\left( k (d\rho^{1/d}+\tau^{1/(8d)}) \right).\]
For the case of the Gaussian distribution the second term does not exist (this terms comes from invariance).
Hence, in the Gaussian setting we can make the error as small as we want by making $\rho$ smaller.
However, in the regular Boolean case, we cannot. So, overall it makes sense to make $\rho  = \poly(\tau)$ (we will determine the exact value later) 
for a total error of (roughly)
\[ O\left( k \tau^{1/(8d)}) \right).\]
This will be our value of $\delta$. Now, the value of $c$ we need to choose is sth like $c = k/\rho.$ 
}
It remains to determine the  corresponding value of $\eps$ in Theorem~\ref{thm:invariance-smooth}.
Note that, by Theorem~\ref{thm:smooth-apx}, the value of the third derivative of the FT-mollified function $\tilde{I}_c$ 
will be at most $(2c)^3 = O(k/\rho)^3$. This is the value of $B$, which
determines the value of $\eps$.
The total error $\eps$ is roughly
\[ \eps = B \cdot \poly(k) \cdot 2^{O(d)} \cdot \sqrt{\tau} = \poly(k) \cdot 2^{O(d)}  \cdot \sqrt{\tau} / \rho^3 =   \poly(k) \cdot 2^{O(d)}  \cdot \tau^{1/8}. \]
\ignore{
Note that we need to have $\sqrt{\tau} / \rho^3 = o(1)$, otherwise the 
error will be too big. E.g. by setting $\rho = \tau^{1/8}$, we get 
that  $\sqrt{\tau} / \rho^3 = \tau^{1/8}$.} Therefore, the total error 
is $\widetilde{\eps} = 2\delta+\eps$ which is at 
most $\poly(k) \cdot 2^{O(d)} \cdot \tau^{1/(8d)}.$
This completes the proof for the case of the AND function. The general case 
follows via a union bound by viewing an arbitrary $k$-junta as 
a disjunction of $2^k$ AND$_k$ functions.

\fi

\section{An algorithmic regularity lemma:  Proof of Lemma \ref{lem:k-reg-alg}} \label{sec:regularity}

\subsection{Useful definitions and tools}

For $p(x_1,\dots,x_n) = \sum_{S \subset [n], |S| \leq d} \widehat{p}(S) \prod_{i \in S} x_i$
a multilinear degree-$d$ polynomial over $\{-1,1\}^n$, recall that
\[
\Inf_i(p) = \sum_{S \ni i} \widehat{p}(S)^2 =
\E_{x_i \in \{-1,1\}}[\Var_{x \setminus x_i \in
\{-1,1\}^{n-1}}[p(x)]]
\]
and that
\begin{equation} \label{eq:var-inf}
\sum_{0 \neq S} \widehat{p}(S)^2 =
\Var[p] \leq \sum_{i=1}^n \Inf_i(p) \leq d \cdot \Var[p].
\end{equation}
We say that $p$ is \emph{$\tau$-regular} if for all $i \in [n]$
we have
\[
\Inf_i(p) \leq \tau \cdot \Var[p].
\]

We will use the following standard tail bound on low-degree polynomials
over $\{-1,1\}^n$, see e.g. Theorem 2.12 of \cite{AH11} for a proof.
(Here and throughout this section unless otherwise indicated, we write $\Pr[\cdot]$, $\E[\cdot]$
and $\Var[\cdot]$ to indicate probability, expectation, and variance
with respect to a uniform draw of $x$ from $\{-1,1\}^n$.)
\begin{theorem}[``degree-$d$ Chernoff bound'',
\cite{AH11}] \label{thm:dcb}
Let $p: \{-1,1\}^n \to \R$ be a degree-$d$ polynomial. For any
$t > e^d$, we have
\[
\Pr[
|p(x) - \E[p]| > t  \cdot \sqrt{\Var[p]} ] \leq
{d e^{-\Omega(t^{2/d})}}.
\]
\ignore{{The same bound holds for $x$ drawn uniformly from $\{-1,1\}^n$
and with $\E_{\{-1,1\}^n}[p]$ and
$\Var_{\{-1,1\}^n}[p]$ in place of $\E_{N(0,1)^n}[p]$ and
$\Var_{N(0,1)}[p]$.}}
\end{theorem}

As a corollary we have:

\begin{corollary} \label{cor:skew} 
There is an absolute constant $C$\ignore{(which we may assume is $>1$; this 
will be convenient later)\rnote{where do we use this?}} such that the 
following holds:

Let $p: \{-1,1\}^n \to \R$ be a degree-$d$ multilinear polynomial that has
\begin{equation} \label{eq:skew}
|\widehat{p}(\emptyset)| =
|\E[p]| \geq (C \log(d /\eps))^{d/2} \cdot
\Var[p].
\end{equation}
Then
$\Pr[\sign(p(x)) \neq \sign(\widehat{p}(\emptyset))]
\leq \eps.$  We say that a polynomial $p$ satisfying
(\ref{eq:skew}) is \emph{$\eps$-skewed}.
\end{corollary}

The following terminology will be convenient for us:
\begin{definition} \label{def:good}
Fix $0 < \eps, \tau < 1/4$ and let $q(x_1,\dots,x_n)$ be a multilinear
degree-$d$ polynomial.  We say that $q$ is
\emph{$(\tau,\eps)$-good} if at least one of the following
two conditions holds:

\begin{enumerate}

\item $q$ is $\tau$-regular; or

\item $q$ is $\eps$-skewed.

\end{enumerate}

\end{definition}

Using this terminology we can give a concise statement of the 
regularity lemma for
a single degree-$d$ polynomial as follows:

\begin{lemma} \label{lem:reg} [regularity lemma, $k=1$]
\cite{DSTW:10,Kane13ccc}
There is a positive absolute constant $A$ such that the following holds:

Let $p$ be a degree-$d$ multilinear polynomial over $\{-1,1\}^n$ and
fix $0 < \tau, \eps, \delta < 1/4$.  Then there is a decision tree
$T$ of depth at most
\[
D_{d,1}(\tau,\eps,\delta) := {\frac 1 \tau} \left( d \log {\frac 1 \tau}
\log {\frac {1} \eps}
\right)^{Ad} \cdot \log {\frac 1 \delta},
\]
\footnote{Throughout the paper we write ``$D_{d,k}(\tau,\eps,\delta)$''
to denote the depth bound of the decision tree given by a regularity lemma
for $k$-tuples of degree-$d$ polynomials in which the regularity
parameter is $\tau$, the skew parameter is $\eps$, and the ``probability
that a leaf is not $(\tau,\eps)$-good'' parameter is $\delta.$}such that
with probability at least $1-\delta$,
at a random leaf $\rho$ the restricted polynomial
$p_\rho$ is $(\tau,\eps)$-good.
\end{lemma}

(We note that \cite{DSTW:10} states the regularity lemma in a form which is
slightly weaker than this because it only claims that
for almost every leaf the restricted PTF at that leaf is
$\tau$-close to $\tau$-regular.
However, inspection of the \cite{DSTW:10} proof shows that it actually
gives the above result:  at almost every leaf the restricted polynomial is either regular
or skewed.
Proposition 15 of \cite{Kane13ccc} gives a statement equivalent to
Lemma \ref{lem:reg} above, along with a streamlined proof.  We further note
that \cite{HKM:09} independently established a very similar regularity lemma,
although with slightly different parameters, that could also be used in
place of Lemma \ref{lem:reg}.)

\ignore{, and follow the
[DSTW10] analysis for our $k$-way regularity lemma.)
However, in the stuff below I am planning
to build on Daniel's analysis instead. It is quite similar to
the [DSTW10] approach but seems slightly simpler since it does not
involve the critical index.)
}

\subsection{The structural result}
The main structural result we prove is the following extension of Lemma \ref{lem:reg} to $k$-tuples
of degree-$d$ polynomials:

\begin{lemma} \label{lem:k-reg} [regularity lemma, general $k$,
general $d$]
Let $p_1,\dots,p_k$ be degree-$d$ multilinear polynomials
over $\{-1,1\}^n$.  Fix $0 < \tau, \eps, \delta < 1/4$.
Then there is a decision tree
$T$ of depth at most
\[
D_{d,k}(\tau,\eps,\delta) \leq 
\left(
{\frac 1 \tau} \cdot \log {\frac 1 \eps}
\right)^{(2d)^{\Theta(k)}}
\cdot \log {\frac 1 \delta}
\]
such that with probability at least $1-\delta$,
at a random leaf $\rho$ all the restricted polynomials $(p_1)_\rho,
\dots,(p_k)_\rho$
are $(\tau,\eps)$-good.
\end{lemma}

\begin{remark} \label{rem:alg}
It is easy to verify (see Theorem 52 of \cite{DDS13:deg2count})
that there is an efficient deterministic algorithm that
constructs the decision tree whose existence is asserted by the
original $k=1$ regularity lemma for degree-$d$ polynomials,
Lemma \ref{lem:reg}.  Given this, inspection of the proof of Lemma
\ref{lem:k-reg} shows that the same is true for the decision tree whose
existence is asserted by Lemma \ref{lem:k-reg}.
(The key observation, in both cases, is that given a degree-$d$
polynomial $q$, it is easy to efficiently deterministically
compute the values $|\widehat{q}(\emptyset)|$, $\Inf_i(q)$
and $\Var[q]$, and thus to determine whether or not $q$ is $\tau$-regular
and whether or not it is $\eps$-skewed.)
Thus in order to establish the algorithmic regularity lemma,
Lemma \ref{lem:k-reg-alg}, it is sufficient to prove Lemma \ref{lem:k-reg}.
\end{remark}

\begin{remark} \label{rem:auto-delta}
Suppose that we prove a result like Lemma \ref{lem:k-reg} but with a bound of
$\gamma(d,k,\tau,\eps,\delta)$ on the RHS upper bounding $D_{d,k}(\tau,\eps,\delta).$
Then it is easy to see that we immediately get a bound of
$\gamma(d,k,\tau,\eps,1/2) \cdot O(\log {\frac 1 \delta})$, simply
by repeating the construction $2 \ln {\frac 1 \delta}$ times
on leaves that do not satisfy the desired $(\tau,\eps)$-good condition.  Thus to prove
Lemma \ref{lem:k-reg} it
suffices to prove a bound of the form 
$\gamma(d,k,\tau,\eps,\delta)$ and indeed this is what
we do below, by showing that \[
\gamma(d,k,\tau,\eps,\delta) =
\left( {\frac 1 \tau} \cdot \log {\frac 1 \eps} \cdot
\log {\frac 1 \delta} \right)^{(2d)^{\Theta(k)}} 
\]
is an upper bound on the solution of the equations 
(\ref{eq:tauprime-epsprime-poly}) and~(\ref{eq:rec-Ddk}) given below; 
see Section \ref{sec:solution}.
\end{remark}

\ignore{
\new{
I did not work out the expression in its full glory but I think
that if $k,d$ are set to constants then it will
just be $\poly({\frac 1 \tau}, \log {\frac 1 \eps})$.
The exponent and coefficients of the polynomial
will depend badly on $k,d$ but I don't think there will be a tower
appearing anywhere. We may be able to improve the dependence
on these parameters $k,d$
by strengthening Claim \ref{claim:junta-reg-stays-reg},
see the red stuff in Section \ref{sec:jrreg}.  Even for $k=d=1$
the correct dependence on $\tau$ is $1/\tau$ and on
$\eps$ is $\log(1/\eps)$ so I don't think we can improve
those dependencies.
}
}

\subsection{Previous results and our approach.} \label{sec:prev-res}

As noted earlier, Gopalan et al. prove a regularity lemma for $k$-tuples of linear
forms in \cite{GOWZ10}.  While their lemma is phrased somewhat differently (they prove
it in a more general setting of product probability spaces), it yields a result that is qualitatively similar
to the special $d=1$ case of Lemma~\ref{lem:k-reg}.  Indeed, the quantitative bound (i.e. the
number of variables that are restricted) in the \cite{GOWZ10} lemma is better than the quantitative bounds we achieve in the case $d=1$.\ignore{; our bound is
${\frac 1 \tau}\left(
\log {\frac 1 \tau} \log {\frac 1 \eps}
\right)^{O(k)}
\log {\frac 1 \delta}
$
and I think theirs would be just
$k \cdot {\frac 1 \tau}\left( \log {\frac 1 \tau} \log {\frac 1 \eps}
\right)^{O(1)} \log {\frac 1 \delta}.$)}  However, there seem to be significant
obstacles in extending the \cite{GOWZ10} approach from linear
forms to degree-$d$ polynomials; we discuss their approach, and contrast it with
our approach, in the rest of this subsection.

The \cite{GOWZ10} regularity lemma works by
``collecting variables'' in a greedy fashion.
Each of the $k$ linear forms has an initial ``budget'' of at most $B$
(the exact value of $B$ is not important for us),
meaning that at most $B$ variables will be restricted ``on its behalf''.
The lemma iteratively builds a set $S$ where each linear form gets to contribute up to $B$ variables
to the set.  At each step in building $S$, if some linear form $\ell_i$ (a) has
not yet exceeded its budget of $B$ variables and (b) is not
yet regular, then a variable that has high influence in $\ell_i$ (relative
to the total influence of all variables in $\ell_i$) is put
into $S$ and the ``budget'' of $\ell_i$ is decreased by one.
If no such linear form exists then the process ends.  It is clear
that the process ends after at most $k B$ variables have been added
into $S$.  At the end of the process, each linear form $\ell_i$
is either regular, or else there have been $B$ occasions when $\ell_i$ contributed
a high-influence variable to $S$.  This ensures that if
$\rho$ is a random restriction fixing the variables in $S$, then
with high probability the restricted $(\ell_i)_{\rho}$ will be
skewed.  (The argument for this goes back to \cite{Servedio:07cc,
DGJ+:10} and employs a simple anti-concentration bound for linear forms with
super-increasing weights.)
\ignore{  Roughly speaking, each of the $B$ occasions when
$\ell_i$ has its influential variable added to $S$ ensure that
the ``tail'' $[n] \setminus S$ of $\ell_i$ will have very small variance
relative to the influences of those $B$ high-influence variables in $\ell_i$. Hence
for a random restriction $\rho$ the ``displacement of the head of
$(\ell_i)_\rho$'' will, with high probability over a random assignment, have
a very large magnitude relative to the variance of the tail, and hence with high
probability the linear form $(\ell_i)_\rho$ will be skewed).}

While these arguments work well for $d=1$ (linear forms), it is not clear
how to extend them to $d>1$.  One issue is that in a linear form,
any restriction of a set $S$ of ``head'' variables leaves the same
``tail'' linear form (changing only the constant term), while this is
not true for higher-degree polynomials.  A more significant obstacle
is that for $d>1$, restricted variables can interact with each other ``in the head'' of the
polynomial $p_i$, and we do not have a degree-$d$ analogue of the
simple anti-concentration bound for linear forms with super-increasing weights that is at the heart of the
$d=1$ argument.  (This anti-concentration bound uses
independence between variables in a linear form to enable a restriction argument
saying that regardless of the existence of other variables ``between'' the variables with
super-increasing weights, a linear form containing super-increasing weights must have
good anti-concentration.  This no longer holds in the higher degree setting.)

\medskip

\noindent {\bf Our approach.}
The idea behind our approach is extremely simple.  Consider first the
case of $k=2$ where there are two polynomials $p_1$ and $p_2$.  For
carefully chosen parameters
$\tau' \ll \tau$ and $\eps' \ll \eps$ we first use the usual regularity lemma
(for a single polynomial) on $p_1$
to construct a decision tree such that at a random leaf $\rho'$, the polynomial
$(p_1)_{\rho'}$ is with high probability $(\tau',\eps')$-good.  Then at each leaf $\rho'$,
we use the usual regularity lemma (for a single polynomial) on $(p_2)_{\rho'}$
to construct a decision tree such that at a random leaf $\rho_2$
of the tree, the polynomial $((p_2)_{\rho'})_{\rho_2}$
is with high probability $(\tau,\eps)$-good.

The only thing that can go wrong in the above scheme is that
$(p_1)_{\rho'}$ is $(\tau',\eps')$-good, but as a result of subsequently
applying the restriction $\rho_2$, the resulting polynomial $((p_1)_{\rho'})_{\rho_2}$
is not $(\tau,\eps)$-good.  However, if
$(p_1)_{\rho'}$ is $\tau'$-regular, then exploiting the fact that $\tau' \ll \tau$,
it can be shown that $((p_1)_{\rho'})_{\rho_2}$ will at least be $\tau$-regular
-- intuitively this is because restricting the (relatively few) variables $\rho_2$ required to ensure
that $(p_2)_{\rho'}$ becomes $(\tau,\eps)$-good, cannot ``damage'' the $\tau'$-regularity
of $(p_1)_{\rho'}$ by too much.
And similarly, if $(p_1)_{\rho'}$ is $\eps'$-skewed, then exploiting the fact that
$\eps' \ll \eps$) it can be shown that $((p_1)_{\rho'})_{\rho_2}$ will at least be $\eps$-skewed, for similar
reasons.
Thus, we can bound the overall failure probability that either polynomial fails to be $(\tau,\eps)$-good
as desired.
The general argument for $k>2$ is an inductive extension of the above
simple argument for $k=2$.
\footnote{As suggested by the sketch given above, we choose $\tau'$ relative to $\tau$ so that
if $(p_1)_{\rho'}$ is $\tau'$-regular then $((p_1)_{\rho'})_{\rho_2}$ will be $\tau$-regular
with probability 1 (and similarly for $\eps'$ and $\eps$).  A natural idea is to weaken
this requirement so that $((p_1)_{\rho'})_{\rho_2}$ will be $\tau$-regular only with high
probability over a random choice of $\rho_2$.  It is possible to give an analysis following
this approach, but the details are significantly more involved and the
resulting overall bound that we were able to obtain is not significantly better than
the bound we achieve with our simpler ``probability-1'' approach.
Very roughly speaking the difficulties arise because it is non-trivial to give a strong tail bound over
the choice of a random restriction sampled from a decision tree in which different sets of variables may be queried on different paths.}

\subsection{Proof of Lemma~\ref{lem:k-reg}}

In this section we prove Lemma~\ref{lem:k-reg}.
The argument is an inductive one using the result for $(k-1)$-tuples
of degree-$d$ polynomials.  As discussed in Remark \ref{rem:auto-delta}, to establish
Lemma \ref{lem:k-reg} it suffices to prove the following:

\begin{lemma} \label{lem:a}
[regularity lemma, general $k$,
general $d>1$]
Let $p_1,\dots,p_k$ be multilinear degree-$d$ polynomials
over $\{-1,1\}^n$.  Fix $0 < \tau, \eps, \delta < 1/4$.
Then there is a decision tree
$T$ of depth at most
\begin{equation}
\label{eq:Ddk}
D_{d,k}(\tau,\eps,\delta) \leq 
\left( {\frac 1 \tau} \cdot \log {\frac 1 \eps} \cdot
\log {\frac 1 \delta} \right)^{(2d)^{\Theta(k)}} 
\end{equation}
such that with probability at least $1-\delta$,
at a random leaf $\rho$ all of $(p_1)_\rho,\dots,(p_k)_\rho$ are
$(\tau,\eps)$-good.
\end{lemma}

\begin{proof}
\ignore{As noted above we assume $d>1.$}
The proof is by induction on $k$.  The base case $k=1$ is given by
Lemma~\ref{lem:reg}; we have that $D_{d,1}(\tau,\eps,\delta)$
satisfies the claimed bound (\ref{eq:Ddk}).
So we may suppose that $k \geq 2$ and that
Lemma~\ref{lem:k-reg} holds for $1,2,\dots,k-1.$

\medskip

Here is a description of how the tree for $p_1,\dots,p_k$ is constructed.

\begin{itemize}

\item [(a)] Let
\begin{equation} \label{eq:tauprime-epsprime-poly}
\tau' =
{\frac {\tau^{\Theta(d)}}
{\left(
d \log {\frac 1 \tau} \log {\frac 1 \eps}
\log {\frac 1 \delta} \right)^{\Theta(d^2)}}}
,
\quad
\eps' = \left(
{\frac \eps d}
\right)^{{\frac 1 {\tau^2}} \left(d \log {\frac 1 \tau} \log {\frac 1 \eps} \right)^{\Theta(d)}
\cdot (\log {\frac 1 \delta})^2}
.
\end{equation}
Let $T'$ be the depth-$D_{d,k-1}(\tau',\eps',\delta/{2})$
decision tree obtained by inductively
applying the ``$k-1$'' case of
Lemma~\ref{lem:a}
to the  polynomials $p_1(x),\dots,p_{k-1}$ with parameters
$\tau'$, $\eps'$, and $\delta/{2}.$

\item [(b)]
For each leaf $\rho'$ in $T'$ such that
all of $(p_1)_{\rho'},\dots,(p_{k-1})_{\rho'}$
are $(\tau',\eps')$-good:

\begin{itemize}

\item Apply
the ``$k=1$'' case of Lemma~\ref{lem:a}
to the  polynomial $(p_k)_{\rho'}$ with parameters
$\tau$, $\eps$, and $\delta/{2}$.  (We say that a leaf/restriction
obtained in this second phase, which we denote
$\rho_k$, \emph{extends $\rho'$}.)

\item
Replace the leaf $\rho'$ with the depth-$D_{d,1}(\tau,\eps,\delta/{2})$
tree (call it $T_{\rho'}$)
thus obtained.

\end{itemize}

\item [(c)]
Output the resulting tree $T$.

\end{itemize}

It is clear that the decision tree $T$ has depth
at most
\begin{equation} \label{eq:rec-Ddk}
D_{d,k}(\tau,\eps,\delta) \eqdef D_{d,k-1}(\tau',\eps',\delta/{2}) +
D_{d,1}(\tau,\eps,\delta/{2}).
\end{equation}
In Section \ref{sec:solution}
we shall show that the quantity $D_{d,k}(\tau,\eps,\delta)$ that is
defined by (\ref{eq:tauprime-epsprime-poly}) and~(\ref{eq:rec-Ddk}) later
indeed satisfies (\ref{eq:Ddk}).

For a given leaf $\rho$ of $T$, let $\rho'$ be the restriction
corresponding to the variables fixed in step (a), and let $\rho_k$ be
the restriction that extends $\rho'$ in step (b), so $\rho = \rho' \rho_k$.

In order for it not to be the case that all of
$(p_1)_{\rho},\dots,(p_k)_{\rho}$ are
$(\tau,\eps)$-good at a leaf $\rho = \rho' \rho_k$,
one of the following must occur:

\begin{enumerate}

\item [(i)]
one of $(p_1)_{\rho'},\dots,(p_{k-1})_{\rho'}$ is not
$(\tau',\eps')$-good;

\item [(ii)]
all of $(p_1)_{\rho'},\dots,(p_{k-1})_{\rho'}$ are
$(\tau',\eps')$-good but
$(p_k)_{\rho'\rho_k}$
is not $(\tau,\eps)$-good;

\item [(iii)]
all of $(p_1)_{\rho'},\dots,(p_{k-1})_{\rho'}$ are
$(\tau',\eps')$-good but
one of
$(p_1)_{\rho'\rho_k},\dots,(p_{k-1})_{\rho'\rho_k}.$
is not $(\tau,\eps)$-good.

\end{enumerate}

By step (a), we have $\Pr[(i)] \leq \delta/{2}$.
Given any fixed $\rho'$ such that
all of $(p_1)_{\rho'},\dots,(p_{k-1})_{\rho'}$
are $(\tau',\eps')$-good, by step (b) we have
$\Pr_{\rho_k}[(p_k)_{\rho'\rho_k}$ is not $(\tau,\eps)$-good$]
\leq \delta/{2}$, and hence
$\Pr[(ii)] \leq \delta/{2}.$
So via a union bound, the desired probability
bound (that with probability $1-\delta$, all of
$(p_1)_{\rho'\rho_k},\dots,(p_k)_{\rho'\rho_k}$
are $(\tau,\eps)$-good at a random leaf $\rho = \rho' \rho_k$)
follows from the following claim, which says that (iii) above cannot occur:

\begin{claim} \label{claim:general-d}
Fix any $i \in \{1,\dots,k-1\}$.  Fix
$\rho'$ to be any leaf in $T'$ such that $(p_i)_{\rho'}$ is
$(\tau',\eps')$-good. Then
$(p_i)_{\rho' \rho_k}$ is $(\tau,\eps)$-good.
\ignore{\[
\Pr_{\rho_k: \rho_k \text{~extends~}\rho'}[
(p_i)_{\rho' \rho_k}
\text{~is not~}(\tau,\eps)\text{-good}]
\leq \delta/{(3k)}.
\]
}
\end{claim}

To prove Claim~\ref{claim:general-d},
let us write $a(x)$ to denote $(p_i)_{\rho'}(x)$, so the polynomial $a$
is $(\tau',\eps')$-good.
There are two cases depending on whether $a$ is
$\tau'$-regular or $\eps'$-skewed.

\medskip

\noindent {\bf Case I:  $a$ is $\tau'$-regular.}
\ignore{We may suppose wlog that $\Var[a]=1.$}
In this case the desired bound is given by the following lemma
which we prove in Section \ref{sec:rsr}.
(Note that the setting of $\tau'$ given in Equation
(\ref{eq:tauprime-epsprime-poly}) is compatible with the setting given
in the lemma below.)

\begin{lemma} \label{lem:reg-stays-reg}
Let $a(x)$ be a degree-$d$ $\tau'$-regular polynomial, where
\[\tau'
=
{\frac 1 2} \left( {\frac {d-1} {eD}} \right)^{d-1} \cdot {\frac 1 {16D^2}}
\quad \text{and~}D=D_{d,1}(\tau,\eps,\delta/2).
\]
Let $T$ be a depth-$D$ decision tree\ignore{ obtained by running the $k=1$ case of
Lemma~\ref{lem:k-reg} on a degree-$d$
polynomial $b(x)$,
using parameters $\tau, \eps, \delta/2$}.
Then for each leaf $\rho$ of $T$,
the polynomial $a_{\rho}$ is $\tau$-regular.
\end{lemma}

\medskip

\noindent {\bf Case II:  $a$ is $\eps'$-skewed.}
In this case the desired bound is given by the following lemma
which we prove in Section \ref{sec:sss}.
(Note that the setting of $\eps'$ given in Equation
(\ref{eq:tauprime-epsprime-poly}) is compatible with the setting given
in the lemma below.)

\begin{lemma} \label{lem:skewed-stays-skewed}
Let $a(x)$ be a degree-$d$ $\eps'$-skewed polynomial, where
\[
\eps' = \left(
{\frac \eps d}
\right)^{\Theta((eD/d)^2)}
\quad \text{and~}D=D_{d,1}(\tau,\eps,\delta/2).
\]
\ignore{
Let $T$ be the depth-$D_{d,1}(\tau,\eps,\delta/2)$
decision tree obtained by running the $k=1$ case of
Lemma~\ref{lem:k-reg} (i.e. Lemma~\ref{lem:reg})
on a degree-$d$ polynomial $b(x)$,
using parameters $\tau, \eps, \delta/2$.
Then for each leaf $\rho$ of $T$, the polynomial $a_{\rho}$ is $\eps$-skewed.
}
Let $T$ be a depth-$D$ decision tree.
Then for each leaf $\rho$ of $T$,
the polynomial $a_{\rho}$ is $\eps$-skewed.
\end{lemma}

These lemmas, together with the argument (given in Section \ref{sec:solution})
showing that $D_{d,k}(\tau,\eps,\delta) = 
\left( {\frac 1 \tau} \cdot \log {\frac 1 \eps} \cdot
\log {\frac 1 \delta} \right)^{(2d)^{\Theta(k)}} 
$
satisfies equations (\ref{eq:tauprime-epsprime-poly}) and (\ref{eq:rec-Ddk}),
yield Claim~\ref{claim:general-d}.
\end{proof}

\subsection{Proof of Lemma~\ref{lem:reg-stays-reg}}
\label{sec:rsr}

\ignore{
Recalling that
$D_{d,1}(\tau,\eps,\delta) = {\frac 1 \tau} \left( d \log {\frac 1 \tau}
\log {\frac {1} \eps}
\right)^{\Theta(d)} \cdot \log {\frac 1 \delta},$
Lemma~\ref{lem:reg-stays-reg} follows directly from the following claim:
}

The key to proving Lemma \ref{lem:reg-stays-reg} is establishing
the following claim.
(Throughout this subsection the expression
``$\left( {\frac {d-1} {es}} \right)^{d-1}$'' and its multiplicative
inverse should both be interpreted as 1 when $d=1.$)

\begin{claim} \label{claim:junta-reg-stays-reg}
Let $p(x_1,\dots,x_n)$ be a multilinear degree-$d$ polynomial which is
$\tau'$-regular.  Let $S \subset [n]$ be a set of at most $s$ variables
and let $\rho$ be a restriction fixing precisely the variables in $S.$
Suppose that
\[\tau' \leq {\frac 1 2} \left( {\frac {d-1} {es}} \right)^{d-1}
\cdot \min \left \{{\frac 1 {16s^2}}, \tau\right\}.
\]
Then\ignore{for any $t > e^d$, with probability at least BLAH over a random
restriction $\rho \in \{-1,1\}^S$,} we have that
$p_\rho$ is $\tau$-regular.
\end{claim}

\noindent {\bf Proof of Claim \ref{claim:junta-reg-stays-reg}:}
\ignore{We may assume without loss of generality that $\Var[p]=\sum_{0 \neq U} \widehat{p}(U)^2 = 1.$}
Since $p$ is $\tau'$-regular, for each $i \in [n]$ we have that $\Inf_i(p)
\leq \tau' \cdot \Var[p]$.  Let $T$ denote $[n] \setminus S$, the set
of variables that ``survive'' the restriction.
The high level idea of the proof is to show that
\ignore{with high probability over a random $\rho$,}both of the following
events take place:

\begin{enumerate}

\item [(i)] No variable $j \in T$ has $\Inf_j(p_\rho)$ ``too much larger''
than $\tau' \cdot \Var[p]$, i.e. all $j \in T$ satisfy
$\Inf_j(p_\rho) \leq \alpha \tau' \Var[p]$ for some
``not too large'' $\alpha > 1$; and

\item [(ii)] The variance $\Var[p_\rho]$ is ``not too much smaller'' than
$\Var[p]$, i.e. $\Var[p_\rho] \geq (1-\beta)\Var[p]$ for some ``not too
large'' $0 < \beta < 1$.
\end{enumerate}

Given (i) and (ii), the definition of regularity implies that $p_\rho$ is
$\left({\frac \alpha {1-\beta}} \cdot  \tau' \right)$-regular.
\ignore{
We show below that we can achieve $\alpha =
\left( {\frac {es} {d-1}} \right)^{d-1}$ and
$\beta = \new{blug}$, which gives
Claim \ref{claim:junta-reg-stays-reg}.
}

\paragraph{Event (i):  Upper bounding influences in the
restricted polynomial.}  We use the following simple
claim, which says that even in the worst case
influences cannot grow too much under restrictions
fixing ``few'' variables in low-degree polynomials.

\begin{claim} \label{claim:inf-cant-grow}
Let $p(x_1,\dots,x_n)$ be a degree-$d$ polynomial
and $S \subset [n]$ a set of at most $s$ variables.  Then for any $j \in
[n] \setminus S$ and any $\rho \in \{-1,1\}^S$, we have
$\Inf_j(p_\rho) \leq \left( {\frac {es} {d-1}} \right)^{d-1}
\cdot \Inf_j(p).$
\end{claim}

\begin{proof}
Let $T$ denote $[n] \setminus S$.  Fix any $j \in T$ and any $U \subseteq T$
such that $j \in U.$ The Fourier coefficient $\widehat{p_\rho}(U)$ equals
$\sum_{S' \subseteq S} \widehat{p}(S' \cup U) \prod_{i \in S'} \rho_i$.
Recalling that $p$ has degree $d$, we see that in order for a subset
$S'$ to make a nonzero contribution to the sum it must be the
case that $|S'| \leq d - |U| \leq d - 1$, so we have that
$\widehat{p_\rho}(U)$ is a $(\pm 1)$-weighted sum of at most
$\sum_{j=0}^{d-1} {s \choose j} \leq \left( {\frac {es} {d-1}} \right)^{d-1}$
Fourier coefficients of $p$.  It follows from Cauchy-Schwarz
that
\[
\widehat{p_\rho}(U)^2 = \left(
\sum_{S' \subseteq S} \widehat{p}(S' \cup U) \prod_{i \in S'} \rho_i
\right)^2 \leq
\left(
\sum_{S' \subseteq S} \widehat{p}(S' \cup U)^2
\right) \cdot
\left( {\frac {es} {d-1}} \right)^{d-1}.
\]
Summing this inequality over all $U \subseteq T$ such that $j \in U$, we get
that
\[
\Inf_j(p_\rho) = \sum_{j \in U \subseteq T} \widehat{p_\rho}(U)^2
\leq
\left( \sum_{j \in V \subseteq [n]} \widehat{p}(V)^2 \right) \cdot
\left( {\frac {es} {d-1}} \right)^{d-1}
=
\left( {\frac {es} {d-1}} \right)^{d-1} \Inf_j(p).
\]
\end{proof}

In the context of event (i), since $\Inf_j(p) \leq \tau' \cdot \Var[p]$, we
get that $\Inf_j(p_\rho) \leq \left( {\frac {es} {d-1}} \right)^{d-1}
\cdot \tau' \cdot \Var[p]$, i.e. the ``$\alpha$'' parameter
of (i) is $\left( {\frac {es} {d-1}} \right)^{d-1}.$

\ignore{
Work out, put in stuff saying influences can't grow too much --
Claim 8 (single var influences don't grow), and general case we need as well.
}

\paragraph{Event (ii):  Lower bounding the variance of the
restricted polynomial.}
The following simple claim says that restricting a single variable in a regular
polynomial cannot decrease the variance by too much:

\begin{claim} \label{claim:1-var-restric}
For $p(x_1,\dots,x_n)$ any multilinear degree-$d$ $\kappa$-regular polynomial
and $\rho$ any restriction that fixes a single variable to a value in $\{-1,1\}$, the
restricted polynomial $p_\rho$ satisfies $\Var[p_\rho] \geq (1 - 2 \sqrt{\kappa})\Var[p].$
\end{claim}
\begin{proof}
Let $\kappa$ be a restriction that fixes $x_1$ to either
$+1$ or $-1$.  For a set $U \subset [n]$, $1 \notin U$ we have that the
sets $U$ and $U \cup \{1\}$ together contribute $\widehat{p}(U)^2 +
\widehat{p}(U \cup \{1\})^2$ to $\Var[p] = \sum_{0 \neq V} \widehat{p}(V)^2.$
In $p_{\rho}$, we have $\widehat{p_{\rho}}(U \cup \{1\})=0$ and
$\widehat{p_{\rho}}(U)=\widehat{p}(U) \pm \widehat{p}(U \cup \{1\})$, so
the sets $U$ and $U \cup \{1\}$ together contribute
$(\widehat{p}(U) \pm \widehat{p}(U \cup \{1\}))^2$ to $\Var[p_{\rho}].$
Hence the difference between the contributions in $p$ versus in $p_{\rho}$
is at most $2|\widehat{p}(U)\widehat{p}(U \cup \{1\})|$ in magnitude.
Summing over all $U \subset [n], 1 \notin U$ we get that
\begin{eqnarray*}
\Var[p] - \Var[p_\rho] &\leq& 2 \sum_{1 \notin U \subset [n]} |\widehat{p}(U) \widehat{p}(U \cup \{1\})| \\
&\leq& 2 \cdot \sqrt{\sum_{1 \notin U \subset [n]} \widehat{p}(U)^2} \cdot
\sqrt{\sum_{1 \notin U \subset [n]} \widehat{p}(U \cup \{1\})^2}\\
&\leq& 2 \cdot \sqrt{\Var[p]} \cdot
\sqrt{\Inf_1(p)}\\
&\leq& 2 \cdot \sqrt{\Var[p]} \cdot \sqrt{\kappa \cdot \Var[p]}
\quad\text{(because $p$ is $\kappa$-regular)}\\
&=& 2 \sqrt{\kappa }\cdot \Var[p].
\end{eqnarray*}
\end{proof}

To establish part (ii), we consider the restriction $\rho$ fixing
all variables in $S$ as being built up by restricting one variable
at a time.  We must be careful in doing this, because the variance lower bound
of Claim \ref{claim:1-var-restric} depends on the regularity of the current
polynomial, and this regularity changes as we successively
restrict variables (indeed this regularity is what we are trying to bound).
Therefore, for $0 \leq t \leq s$, let us define $\reg_t$ as the ``worst-case''
(largest possible)
regularity of the polynomial $p$ after $t$ variables have been restricted
(so we have $\reg_0=\tau'$ since by assumption $p$ is initially
$\tau'$-regular); our goal is to upper bound $\reg_s.$
For $0 \leq t \leq s$, let $\rho_t$ denote a restriction that fixes
exactly $t$ of the $s$ variables in $S$ (so $p_{\rho_0}$ is simply $p$).
By repeated applications of Claim \ref{claim:1-var-restric} we have

\begin{eqnarray*}
\Var[p_{\rho_t}] &\geq&
\left(1 - 2 \sqrt{\reg_{t-1}}\right) \Var[p_{\rho_{t-1}} ]\\
&\geq& \left(1 - 2 \sqrt{\reg_{t-1}}\right)
\left(1 - 2 \sqrt{\reg_{t-2}}\right) \Var[p_{\rho_{t-2}} ]\\
&\geq& \cdots\\
&\geq&
(1 - 2 \sqrt{\reg_{t-1}}) \cdots (1 - 2 \sqrt{\reg_0})
\Var[p],
\end{eqnarray*}
and by Claim \ref{claim:inf-cant-grow} we have that every $j$ satisfies
$\Inf_j(p_{\rho_t}) \leq
\left( {\frac {es} {d-1}} \right)^{d-1}\cdot \max_{i \in [n]} \Inf_i(p).$
We shall set parameters so that $\sum_{r=0}^{s-1} \sqrt{\reg_r} \leq
{\frac 1 4}$; since
\[
(1 - 2 \sqrt{\reg_{t-1}}) \cdots (1 - 2 \sqrt{\reg_0}) \geq
1 - 2 \sum_{r=0}^{s-1} \sqrt{\reg_r},
\]
this means that for all $0 \leq t \leq s$ we shall have $\Var[p_{\rho_t}]
\geq {\frac 1 2} \Var[p].$  We therefore have that every $t$ satisfies
\[
{\frac {\Inf_{j}(p_{\rho_t})}
{\Var[p_{\rho_t}]}}
\leq
{\frac {
\left( {\frac {es} {d-1}} \right)^{d-1}
\cdot \max_{i \in [n]} \Inf_i(p)
}
{{\frac 1 2} \Var[p]}
}
\leq 2 \left( {\frac {es} {d-1}} \right)^{d-1} \tau',
\]
and therefore $\reg_{t} \leq 2\left( {\frac {es} {d-1}} \right)^{d-1}
\tau'.$   Finally, to confirm that $\sum_{r=0}^{s-1} \sqrt{\reg_r} \leq
{\frac 1 4}$ as required, we observe that we have
\[
\sum_{r=0}^{s-1} \sqrt{\reg_r}
\leq s \sqrt{\max_{0 \leq r \leq s-1} \reg_r} \leq
s \sqrt{2 \left( {\frac {es} {d-1}} \right)^{d-1} \tau'}
\]
which is at most ${\frac 1 4}$ by the conditions that Claim
\ref{claim:junta-reg-stays-reg} puts on $\tau'$.  So we indeed have that
\[
\reg_s \leq 2\left( {\frac {es} {d-1}} \right)^{d-1}
\tau' \leq \tau,
\]
again by the conditions that Claim \ref{claim:junta-reg-stays-reg}
puts on $\tau'.$  This concludes the proof of
Claim \ref{claim:junta-reg-stays-reg}.
\qed

\ignore{
Since Claim \ref{claim:junta-reg-stays-reg} holds for all restrictions fixing
at most $s$ variables rather than just with high probability, it immediately
gives the following consequence for
``decision tree restrictions.''

\begin{claim} \label{claim:tree-reg-stays-reg}
Let $p(x_1,\dots,x_n)$ be a multilinear degree-$d$ polynomial which is
$\tau'$-regular.  Let $T$ be a decision tree of depth $D$, and suppose
that
\[\tau'
\leq {\frac 1 2} \left( {\frac {d-1} {eD}} \right)^{d-1} \cdot \min \left \{{\frac 1 {16D^2}},
 \tau\right\}.
\]
Then\ignore{for any
$t > e^{d}$, with probability at least $1 - 3 \ell d e^{-\Omega(t^{1/d})}$
over a random restriction $\rho$ corresponding to a random leaf of $T$,}
for a restriction $\rho$ corresponding to a leaf of $T$, we have
that $p_\rho$ is $\tau$-regular.
\end{claim}
}

With Claim \ref{claim:junta-reg-stays-reg} in hand
we are ready to prove Lemma \ref{lem:reg-stays-reg}.
As stated in the lemma, let
$a(x)$ be a degree-$d$ $\tau'$-regular polynomial, where
\[\tau'
= {\frac 1 2} \left( {\frac {d-1} {eD}} \right)^{d-1} \cdot {\frac 1 {16D^2}}
\quad \text{and~}D=D_{d,1}(\tau,\eps,\delta/2)
\]
(note that by the definition of the $D_{d,1}(\cdot,\cdot,\cdot)$
function we have that ${\frac 1 {16D^2}} < \tau$). \ignore{
Let $T$ be the depth-$D$
decision tree obtained by running the $k=1$ case of
Lemma~\ref{lem:k-reg} on a degree-$d$
polynomial $b(x)$,
using parameters $\tau, \eps, \delta/3$.}
Claim \ref{claim:junta-reg-stays-reg}
gives that
\ignore{with probability at least $1 - {\frac \delta {3k}}$ over a
random leaf $\rho$ of $T$,}
at every leaf $\rho$ of $T$
the polynomial $a_{\rho}$ is $\tau$-regular, and
Lemma \ref{lem:reg-stays-reg} is proved.
\qed

\subsection{Proof of Lemma~\ref{lem:skewed-stays-skewed}} \label{sec:sss}

We may suppose w.l.o.g. that $\Var[a] = 1$.
Since $a$ is $\eps'$-skewed, we may suppose that
$\widehat{p}(\emptyset) \geq (C \log (d/\eps'))^{d/2}$.

Let $\rho$ be any restriction fixing up to $D$ variables.  The idea of the proof
is to show that  (i) $\widehat{p_\rho}(\emptyset)>0$ is still ``fairly large'', and
(ii) $\Var[p]$ is ``not too large''; together these conditions imply that $p_\rho$ is
skewed.  We get both (i) and (ii) from the following claim which is quite similar to
Claim \ref{claim:inf-cant-grow}:

\begin{claim} \label{claim:gutsofsss}
Let $p(x_1,\dots,x_n)$ be a degree-$d$ polynomial with $\Var[p]=1$ and $\widehat{p}(\emptyset)=0.$
Let $S \subset [n]$ be a set of at most $s$ variables.  Then
for any $\rho \in \{-1,1\}^{S}$, we have that (i) $|\widehat{p_\rho}(\emptyset)| \leq \left(
{\frac {es} d}\right)^{d/2}$, and (ii) $\Var[p_\rho] \leq \left( {\frac {es}{d-1}} \right)^{d-1}
\Var[p].$
\end{claim}

\begin{proof}
Let $T$ denote $[n] \setminus S$ and let us write $x_S$ to denote the vector
of variables $(x_i)_{i \in S}$ and likewise $x_T$ denotes $(x_i)_{i \in T}.$
We may write $p(x)$ as $p(x_S,x_T) = p'(x_S) + q(x_S,x_T)$ where
$p'(x_S)$ is the truncation of $p$ comprising only the monomials all of whose
variables are in $S$, i.e. $p'(x_S) = \sum_{U \subseteq S} \widehat{p}(U)
\prod_{i \in U} x_i.$

For part (i), it is clear that for $\rho \in \{-1,1\}^S$
we have that $\widehat{p_\rho}$ equals $p'(\rho)$.  Since
\[
\Var[p'] =
\sum_{U \subseteq S} \widehat{p}(U)^2
\leq
\sum_{U \subseteq [n]} \widehat{p}(U)^2
= \Var[p] = 1,
\]
we have
\[
\left|\widehat{p_\rho}(\emptyset)\right|
=
\left|
\sum_{U \subseteq S} \widehat{p}(U) \prod_{i \in U} \rho_i
\right|
\leq
\sqrt{\sum_{U \subseteq S}
\widehat{p}(U)^2} \cdot \sqrt{\sum_{j=0}^d {s \choose j}} \leq {{es} \choose d}^{d/2}.
\]

For (ii), as in the proof of Claim \ref{claim:inf-cant-grow} we get that any
nonempty $U \subseteq T$ has

\[
\widehat{p_\rho}(U)^2 = \left(
\sum_{S' \subseteq S} \widehat{p}(S' \cup U) \prod_{i \in S'} \rho_i
\right)^2 \leq
\left(
\sum_{S' \subseteq S} \widehat{p}(S' \cup U)^2
\right) \cdot
\left( {\frac {es} {d-1}} \right)^{d-1}.
\]

Summing this inequality over all nonempty $U \subseteq T$,we get
that
\[
\Var[p_\rho] = \sum_{\emptyset \neq U \subseteq T} \widehat{p_\rho}(U)^2
\leq
\left( \sum_{\emptyset \neq V \subseteq [n]} \widehat{p}(V)^2 \right) \cdot
\left( {\frac {es} {d-1}} \right)^{d-1}
=
\left( {\frac {es} {d-1}} \right)^{d-1} \Var[p].
\]
This concludes the proof of Claim \ref{claim:gutsofsss}.
\end{proof}

\noindent {\bf Proof of Lemma~\ref{lem:skewed-stays-skewed}:}
Fix any leaf $\rho$ in the decision tree $T$ from the statement of
Lemma~\ref{lem:skewed-stays-skewed}.  As noted at the start of this subsection
we may suppose w.l.o.g. that $\Var[a]=1$ and 
$\widehat{a}(\emptyset) \geq
(C \log(d /\eps'))^{d/2}$.  Claim \ref{claim:gutsofsss} gives us that
$\widehat{p_\rho}(\emptyset) \geq 
(C \log ({\frac d {\eps'}}))^{d/2} - ({\frac {eD}d})^{d/2}$ and that
$\Var[p_\rho] \leq ({\frac {eD}{d-1}})^{d-1}$, so $p_\rho$ must be $\eps$-skewed as long
as the following inequality holds:

\begin{equation}
\label{eq:c2}
\left(C \log \left({\frac d {\eps'}}\right)\right)^{d/2} \geq \left({\frac {eD}d}\right)^{d/2} + 
\left({\frac {eD}{d-1}}\right)^{d-1} \cdot
\left(C \log \left({\frac d {\eps}}\right)\right)^{d/2}.
\end{equation}

Simplifying the above inequality we find that taking $\eps'$ as specified in 
Lemma~\ref{lem:skewed-stays-skewed} satisfies the inequality, and
Lemma~\ref{lem:skewed-stays-skewed} is proved.

\subsection{The solution to the equations}
\label{sec:solution}

To complete the proof of Lemma~\ref{lem:k-reg} it suffices to show that the 
quantity $D_{d,k}(\tau,\eps,\delta)$ that is
defined by (\ref{eq:tauprime-epsprime-poly}) and~(\ref{eq:rec-Ddk})
indeed satisfies (\ref{eq:Ddk}).
It is clear from Lemma \ref{lem:reg} that (\ref{eq:Ddk}) holds when
$k=1$.  A tedious but straightforward induction using 
(\ref{eq:tauprime-epsprime-poly}) and~(\ref{eq:rec-Ddk})
shows that (\ref{eq:Ddk}) gives a valid upper bound.  (To verify
the inductive step it is helpful to
note that for $k >1$, equations
(\ref{eq:tauprime-epsprime-poly}) and~(\ref{eq:rec-Ddk})
together imply that 
$D_{d,k}(\tau,\eps,\delta) \leq 2 D_{d,k-1}(\tau',\eps',\delta/2)$.)


\bibliography{allrefs}
\bibliographystyle{alpha}


\end{document}